\documentclass[11pt]{article}
    
    \usepackage[utf8]{inputenc}
\usepackage[english]{babel}
\usepackage{amsfonts}
\usepackage{amssymb}
\usepackage{enumitem}
\usepackage{bbm}
\usepackage{mathtools}
\usepackage{subcaption}
\usepackage{booktabs}
\usepackage{setspace}
\usepackage{graphicx}
\usepackage{fullpage}

\usepackage{soul}
\usepackage[hidelinks]{hyperref} % get rid of ugly boxes on references
\usepackage{natbib}
\usepackage{amsthm}
\usepackage{longtable}
\usepackage{rotating}
\usepackage{caption}
\usepackage{xcolor}
\usepackage{lscape}

    \renewcommand{\P}{\mathop{}\!\textnormal{P}}
    \newcommand{\sP}{\mathop{}\!\mathbb{P}_n}
    \newcommand{\E}{\mathop{}\!\textnormal{E}}
    \newcommand{\sE}{\mathop{}\!\mathbb{E}_n}
    \newcommand{\Cov}{\mathop{}\!\textnormal{Cov}}
    \newcommand{\sCov}{\mathop{}\!\mathbb{C}\textnormal{ov}_n}
    \newcommand{\Var}{\mathop{}\!\textnormal{Var}}
    \newcommand{\sVar}{\mathop{}\!\mathbb{V}\textnormal{ar}_n}

    \newcommand{\N}{\mathcal{N}}

    \newcommand{\Oh}{\mathop{}\!\mathcal{O}_{\P}}
    
    \newcommand{\cd}{\stackrel{d}{\longrightarrow}}
    \newcommand{\cp}{\stackrel{p}{\longrightarrow}}

    \newcommand{\0}{\mathbf{0}}

    \newcommand{\Ind}{\mathbbm{1}}  % Indicator
    \newcommand{\peq}{\mathrel{\phantom{=}}}

    \newcommand{\orth}{\perp}

    \newcommand{\Z}{\mathcal{Z}}  % Instrument space
    \newcommand{\err}{\varepsilon}  % Outcome error
    \newtheorem{proposition}{Proposition}
    
    \newtheorem{corollary}{Corollary}

    \newtheorem{lemma}{Lemma}

    \addto\extrasenglish{%

    }

    \makeatletter
    \def\blfootnote{\xdef\@thefnmark{}\@footnotetext}
    \makeatother
    \usepackage[hang,flushmargin]{footmisc}

    \usepackage[colorinlistoftodos,color=yellow!20!white]{todonotes}
	\presetkeys{todonotes}{inline}{}

	\usepackage{marginnote}
	
	\usepackage[capposition=top, margins=centering]{floatrow}

    \usepackage{versions}
    \includeversion{graphstables}
    \includeversion{vappendix}

	\author{Stephen Coussens \\ Columbia University \and Jann Spiess \\ Stanford University}
    \title{
		Improving Inference from Simple Instruments\\
        through Compliance Estimation
    }

    \date{ August 6, 2021 \\
            \vspace{25pt}
            \textit{Comments welcome!}
            }

\begin{document}

    \maketitle
    \begin{abstract}
    Instrumental variables (IV) regression is widely used to estimate causal treatment effects in settings where receipt of treatment is not fully random, but there exists an instrument that generates exogenous variation in treatment exposure. 
	While IV can recover consistent treatment effect estimates, they are often noisy.
    Building upon earlier work in biostatistics \citep{Joffe2003-uq} and relating to an evolving literature in econometrics \citep[including][]{abadie2019instrumental,Huntington-Klein,Borusyak2020-xu}, we study how to improve the efficiency of IV estimates by exploiting the predictable variation in the strength of the instrument.
    In the case where both the treatment and instrument are binary and the instrument is independent of baseline covariates,
    we study weighting each observation according to its estimated compliance (that is, its conditional probability of being affected by the instrument), which we motivate from a (constrained) solution of the first-stage prediction problem implicit to IV.
    The resulting estimator can leverage machine learning to estimate compliance as a function of baseline covariates.
    We derive the large-sample properties of a specific implementation of a weighted IV estimator in the potential outcomes and local average treatment effect (LATE) frameworks,
    and provide tools for inference that remain valid even when the weights are estimated nonparametrically.
    With both theoretical results and a simulation study, we demonstrate that compliance weighting meaningfully reduces the variance of IV estimates when first-stage heterogeneity is present, and that this improvement often outweighs any difference between the compliance-weighted and unweighted IV estimands.
    These results suggest that in a variety of applied settings, the precision of IV estimates can be substantially improved by incorporating compliance estimation.
\end{abstract}
    \blfootnote{Stephen Coussens, Mailman School of Public Health, Columbia University, \href{mailto:sc4501@cumc.columbia.edu}{sc4501@cumc.columbia.edu}. Jann Spiess, Graduate School of Business, Stanford University, \href{mailto:jspiess@stanford.edu}{jspiess@stanford.edu}. We thank Kirill Borusyak, Amanda Kowalski, Greg Lewis, and Ashesh Rambachan for helpful comments.}

    \newpage

    \section{Introduction}

Instrumental variables (IV) regression can recover consistent estimates of treatment effects in quasi-experimental studies and in randomized controlled trials (RCTs) where treatment randomization is infeasible, but there is at least some exogenous variation in treatment exposure generated by an instrument. Yet in practice, IV estimates are often noisy. With this in mind, in this paper we discuss how the precision of IV estimates can be improved when there is predictable variation in the strength of the instrument, provide valid inference for a specific implementation that can leverage machine learning algorithms, and document potential precision improvements in a simulation study.

IV estimation frequently involves the use of a two-stage estimator: the first stage predicts variation in treatment exposure using the instrumental variable(s), while the second stage regresses the outcome of interest on the first-stage predicted values.
Better prediction of exogenous variation in treatment exposure can therefore improve the precision of treatment effect estimates.
The proliferation of powerful machine learning algorithms designed to make predictions from high-dimensional data has thus spawned efforts to improve IV precision by combining a large number of relatively weak instruments into a single, strong instrument \citep[e.g.][]{Belloni2012-nu,Hansen2014-wy}. Yet in practice, valid instruments are typically hard to find, so applied research is frequently conducted with only a single simple instrument that takes on a limited number of values -- such as when randomized assignment to either a treatment group or a control group is used as a binary instrument for treatment.%
\footnote{
	As we explore in this paper, even when only a single simple instrument is available, we can often create many more by interacting this instrument with covariates. Methods for high-dimensional instruments therefore remain applicable even in the simple-instrument case.
	In contrast to the literature on high-dimensional instruments, we explicitly consider the specific structure of instruments created by the interaction of a simple instrument with a function of covariates.
}

We consider an IV estimator for simple instruments that improves efficiency by weighting each observation according to its estimated conditional probability of being a complier (that is, its probability of being affected by the instrument) in the case where the instrument is independent of baseline covariates.
The underlying idea of weighting observations in instrumental variables estimation according to their estimated compliance is not new -- in the biostatistics literature, it can be traced back to at least \cite{Joffe2003-uq}. Moreover, the implementation that we consider is close to one considered by \cite{Huntington-Klein} as a tool to reduce finite-sample bias, and is similar to proposals by \cite{abadie2019instrumental} and \cite{Borusyak2020-xu} in linear IV models.
However, compliance-weighted estimators have seen very limited adoption in applied research.%
\footnote{
	Perhaps the most closely related applied research can be found in \citet{einav2018long}, which reduces the variance of IV estimates by allowing the first stage coefficient to vary across quintiles of estimated compliance.} 

In this paper, we aim to fill gaps between the idea of compliance weighting and its implementation by
showing how compliance weighting arises naturally from the (constrained) solution of the first-stage prediction problem implicit to IV;
by deriving the large-sample properties of a weighted IV estimator in the familiar potential outcomes \citep{rubin1977assignment} and local average treatment effect (LATE) \citep{imbens1994identification} frameworks; 
and by providing tools for valid inference that remain applicable when the first stage is estimated nonparametrically, allowing the researcher to leverage machine learning techniques to flexibly estimate compliance as a function of potentially high-dimensional covariates.

We begin by pointing out that when the instrument is simple (specifically, binary), 
then constructing efficient instruments by interacting them with covariates represents a prediction problem under constraints.
Namely, we want to predict treatment from instruments and covariates, while also ensuring that the instrument remains valid.
In recent work, \citet{Borusyak2020-xu} tackle this challenge by  recentering an unconstrained prediction of treatment from instruments and covariates so that only the variation in treatment that is correlated with the instrument remains.
We instead characterize interacted instruments that are valid by construction and do not require recentering, provided that the covariates and the instrument are independent.
We show that choosing instruments from this constrained class is equivalent to choosing weight functions that depend on covariates only.
This approach is particularly applicable in high-dimensional data or when using nonparametric estimation techniques, in which case recovering valid instruments from more complex unconstrained predictions may be inefficient.

Next, we show that finding the efficient weights implicit in the IV estimator with first-stage interactions requires solving a heterogeneous treatment effect estimation problem.
In a homogeneous and homoscedastic baseline case, the efficient weight for IV estimation under first-stage monotonicity is known to be the probability of compliance \citep{Joffe2003-uq}, which is the causal effect of the instrument on treatment.
Solving for the efficient instrument in this case is therefore equivalent to estimating the conditional average treatment effect of the instrument on treatment (i.e., the conditional probability of compliance).

In order to briefly illustrate the important role that compliance estimation plays in reducing the variance of IV estimators,
consider a simple regression on a sample of size $N$ with binary instrument and treatment, where treatment effects are homogeneous, variance is homoscedastic, and the average probability of compliance (i.e. the first-stage coefficient) is $\alpha < 1$.
In this case, the variance of the standard two-stage least-squares (2SLS) estimator is $\frac{\sigma^2}{N p(1-p) \alpha^2}$, where $\sigma^2$ is the variance of the regression residual in the reduced-form equation that relates the outcome to the instrument, and $p$ is the instrument mean.
Now suppose that it is possible to perfectly identify all $\alpha N$ compliers in this sample and estimate the treatment effect using this smaller yet perfectly compliant sub-sample. Under this hypothetical scenario, the variance of the 2SLS estimator would fall to $\frac{\sigma^2}{N p(1-p) \alpha}$, a sizable reduction that would be equivalent to increasing the imperfectly compliant sample size from $N$ to $\frac{N}{\alpha}$.
While in practice we will not generally be able to perfectly identify compliance -- an attribute that cannot be directly observed -- we can nonetheless estimate the probability of compliance conditional on covariates \citep{follmann2000effect}.

Having argued in favor of compliance-weighting in a simple baseline case, we relax the assumptions of homoscedasticity and treatment effect homogeneity, and characterize the asymptotic distribution of a weighted IV estimator in this more general setting. We show that weighting the IV estimator changes the estimand when treatment effects are heterogeneous, and alters the asymptotic variance under heteroscedasticity.  We point out that the weighted estimand can be interpreted as a convex combination of conditional LATEs, and then characterize the conditions under which this weighted IV estimator improves precision relative to the unweighted IV estimator, even when the unweighted unconditional LATE is the parameter of interest. We also consider an extension of this compliance-weighting approach that can further improve precision by also accounting for second-stage heteroscedasticity.

We then discuss an implementation of compliance-weighted IV estimation using machine learning. 
If the first stage can be nonparametric, then the causal problem of estimating compliance is different from an unconstrained prediction problem, in that naive predictions from covariates and the instrument would yield inefficient estimates of treatment effects.
Like \cite{Huntington-Klein}, we therefore study the use of recent tools from the intersection of causal inference and machine learning that are specifically designed to estimate conditional average treatment effects.
Our implementation uses a cross-fitting scheme that repeatedly splits the sample to avoid biases from overfitting and enable valid inference even for high-dimensional estimation tools.

We go on to address the issue of inference when weights are estimated through such nonparametric procedures, which is a non-standard problem since the estimated weights affect the asymptotic distribution.
In order to keep the assumptions placed on nonparametric first-stage estimation to a minimum, our framework considers estimation relative to an empirically-weighted estimand that varies with each sample.
In this framework, we demonstrate that the estimator is asymptotically Normal, and that commonly-used approaches to estimating robust standard errors will provide for valid inference.
These results stem from a general result on the convergence of empirically weighted sums that relies on sample splitting and may be of independent interest.

In simulations, we demonstrate that augmenting simple instruments through compliance estimation meaningfully reduces the variance of IV estimates when first-stage heterogeneity is present, and that these improvements typically outweigh biases relative to the unweighted IV estimand or ineffciencies from heteroscedastic errors.
We present a simulation study that systematically varies treatment effects, residual errors, first-stage heterogeneity, and weight estimation methods, comparing the performance of compliance-weighted IV to that of a baseline 2SLS estimator. Our results suggest that in a variety of scenarios, applied researchers stand to substantially improve the precision of their IV estimates through the use of compliance weights.

\paragraph{Related literature.}
This work is most closely related to research in biostatistics and econometrics that proposes weighting or selection by compliance \citep{Joffe2003-uq,coussens2018,abadie2019instrumental,Huntington-Klein}, interacts instruments with covariates \citep[e.g.][]{Bond2007-fd,Borusyak2020-xu}, employs complier classification to analyze and improve the external validity of IV estimates \citep{Kennedy2018-fe}, or explicitly discusses the role of exogenous covariates in linear IV models \citep{Chen2020-xg}.

\citet{Joffe2003-uq} propose weighting observations in an experiment according to their probability of compliance. \citet{abadie2019instrumental} discard groups of observations with low estimated compliance, and employ an interacted IV estimator among the remaining groups that implicitly weights each group according to compliance. Their approach considers a grouped linear simultaneous equations model with an increasing number of groups and a higher-order expansion of second-stage mean-squared error.
\citet{Huntington-Klein} independently considers a variety of compliance grouping and weighting schemes in an effort to tackle a different problem: the small-sample bias of IV estimators in a linear model with heterogeneous first- and second-stage effects.
Like us, \citet{Huntington-Klein} characterizes estimands in terms of weighted causal effects, and uses the causal forest \citep{wager2018estimation} to estimate compliance weights.
The estimator that we consider is similar to one of \citeauthor{Huntington-Klein}'s~(\citeyear{Huntington-Klein}) proposed implementations, but differs primarily in terms of sample splitting, which allows us to derive its large-sample properties and provide formal justifications for valid inference.

Our approach also relates to recent, independent work by \citet{Borusyak2020-xu}, who consider non-random exposure to exogenous variation.
Their work tackles the more complex task of parametrically constructing valid recentered instruments when exposure is itself non-random.
Despite these different approaches, their construction based on \citet{chamberlain1987asymptotic} effectively yields a similar efficient, compliance-weighted IV estimator when considered under our stronger exogeneity assumptions, and the authors also offer a supplemental interpretation in a potential outcomes framework similar to our setup in \autoref{sec:generalbinary}.

Finally, \citet{Chen2020-xg} independently consider the role of exogenous covariates in the construction of instruments using machine learning, and provide tools for inference that remain applicable in that case. Also invoking the characterization of efficient instruments in \citet{chamberlain1987asymptotic}, \citet{Chen2020-xg} use a cross-fitting construction to implement the optimal instrument, and characterize estimands in terms of weighted averages of marginal treatment effects.
Unlike us, the authors do not assume that covariates and instrumental variables are independent, and specifically tackle complications that arise from the presence of covariates in that case.

\paragraph{Contribution.}
Relative to existing work on leveraging compliance information, we formally motivate compliance weighting as the solution to a constrained, fully nonparametric first-stage prediction problem.
We then interpret the estimand and analyze the properties of the estimator in a potential outcomes framework that does not impose any linearity assumptions.
As our main technical contribution,
we derive large-sample limits and propose inference on weighted local average treatment effects for a specific cross-fitting implementation that remain valid with arbitrary treatment effect heterogeneity and when the first-stage estimation of compliance is fully nonparametric.

Despite methodological advances in leveraging first-stage compliance information, we are not aware of widespread adoption of compliance-weighting approaches in applied research. 
Instead, attempts to harness first-stage heterogeneity have typically been limited to ad-hoc analysis of subgroups with higher estimated compliance, a practice that tends to result in inflated IV estimates and $t$-statistics \citep{abadie2019instrumental}.
We hope that we can contribute to the adoption of such methods by providing valid inference when using estimated compliance in IV models.

\section{IV Estimation with Oracle Augmented Instruments}
\label{sec:oracle}

How can we exploit covariates to improve instrumental variable estimation using only a single, simple instrument? To fix ideas, consider the simple homogeneous linear instrumental variables model
\begin{align}
    \label{eqn:baselinesimple}
    Y &= \psi + D \tau + \varepsilon
    &
    \E[\varepsilon] &= 0
    &
    Z \orth \varepsilon, X
\end{align}
where both the outcome $Y$ and the treatment variable $D$ are real, univariate random variables, and the univariate, real-valued instrument $Z$ is independent of the error term $\varepsilon$ and the regressors $X$.\footnote{Although independence is a strong assumption, this assumption is often fulfilled by design, e.g., using random assignment in an RCT as an instrument for treatment.} Since $\Cov(\varepsilon,Z) = 0$, if $\Cov(D,Z) \neq 0$ then the standard IV estimator
\begin{align*}
    \hat{\tau}^{\text{IV}} = \frac{\sCov(Y,Z)}{\sCov(Y,D)},
\end{align*}
where we write $\sCov(\cdot,\cdot)$ for the empirical covariance in an iid random sample of size $n$ from that distribution,
is a $\sqrt{n}$-consistent and asymptotically normal estimator of $\tau$ under standard regularity conditions. In effect, this estimator divides the ``reduced-form'' OLS coefficient in the linear regression of $Y$ on $Z$ by the ``first-stage'' OLS coefficient in the linear regression of $D$ on $Z$; it is also equivalent to the 2SLS estimator that first regresses $D$ on $Z$ to form predicted values $\hat{D}$ and then regresses $Y$ on $\hat{D}$. However, this estimator does not make use of covariates $X$.

One commonly practiced way to incorporate covariates $X$ into IV estimation is as control variables in order to reduce the residual variation in $Y$. For example, if we assume that
\begin{align*}
    Y &= \psi_2 + D \tau + X'\varphi_2 + \varepsilon_2,
    &
    D &= \psi_1 + Z \pi + X'\varphi_1 + \varepsilon_1,
\end{align*}
then we can reduce the variance of the IV estimator by residualizing with respect to $X$ first. Given the exogeneity of $Z$, this is not required for consistency, but can nonetheless improve efficiency.

In this paper, we pursue a different, complementary use of pre-assignment covariates: in addition to controlling for covariates, we can also use them to construct new instruments by combining an existing instrument with these variables.
In \autoref{sec:augmenting}, we relate such interacted instruments to the solution of a constrained prediction problem in a homogeneous and homoscedastic baseline case, which we solve for binary instruments.
In \autoref{sec:weighted}, we interpret the resulting estimator as a weighting estimator.
Since in many applications researchers may not wish to assume homogeneous treatment effects or homoscedasticity, in \autoref{sec:generalbinary} we discard these assumptions and discuss the properties of this weighted IV estimator in a potential outcomes framework.
In \autoref{sec:generallinear}, we discuss extensions of these results to non-binary instruments.

\subsection{Augmenting simple instruments with covariates}
\label{sec:augmenting}

In this section, we consider the construction of augmented instruments by interacting them with functions of pre-assignment covariates \citep[e.g.][]{Bond2007-fd}. When covariates $X$ and error term $\varepsilon$ are independent of the instrument as in \autoref{eqn:baselinesimple}, then there exist transformations of the form $h(Z,X)$ that also fulfill $\Cov(\varepsilon,h(Z,X)) = 0$ and can thus serve as valid instruments. 
This suggests the estimator
\begin{align*}  
    \hat{\tau}^h = \frac{\sCov(Y,h(Z,X))}{\sCov(D,h(Z,X))}.
\end{align*}
This augmented IV estimator $\hat{\tau}^h$, estimated from an iid sample, is consistent for the treatment effect $\tau$ if $\Cov(D,h(Z,X)) \neq 0$ holds and some mild regularity conditions guarantee the existence of sufficient population moments.

How would one choose among instruments $h(Z,X)$?
To fix ideas, we make the additional homogeneity and homoscedasticity assumption that the residuals in \autoref{eqn:baselinesimple} a.s.\ fulfill $\E[\varepsilon^2|X] = \sigma^2$. This assumption (which we drop in \autoref{sec:generalbinary} below) is restrictive -- it limits treatment effects, conditional baseline expectations, and residuals.
Under this assumption, 
the asymptotic distribution of $\hat{\tau}^h$ is
\begin{align*}
    \sqrt{n} (\hat{\tau}^h - \tau)
    &\cd
    \N\left(0,V^h\right),
    &
    V^h &= \sigma^2 \frac{\Var(h(Z,X))}{\Cov^2(D, h(Z,X))},
\end{align*}
as $n \rightarrow \infty$.
The asymptotic variance in this case has a particularly intuitive interpretation.
Indeed, if we normalize $h(Z,X)$ to $\E[h(Z,X)] = \E[D]$ and $\Var(h(Z,X)) = \Cov(D,h(Z,X))$ 
--
which is without loss since $\hat{\tau}^h$ and thus also $V^h$ are invariant to shift and scale transformations%
\footnote{Formally, we could define   
    \begin{align*}
        h^*(Z,X) &= \E[D] + \frac{h(Z,X) - \E[h(Z,X)]}{\Var(h(Z,X))} \Cov(h(Z,X), D),
        &
        R^2_h &= 1 - \frac{\E[(D - h^*(Z,X))^2]}{\Var(D)},
    \end{align*}
    which is the best linear predictor of $D$ from $h(Z,X)$.
}
--
then
\begin{align}
    \label{eqn:predictionanalogy}
    V^h
    &= \frac{\sigma^2}{\Var(D) \: R_h^2},
    &
    R_h^2
    &= 1 - \frac{\E[(D - h(Z,X))^2]}{\Var(D)},
\end{align}
where $R_h^2$ is the (population) coefficient of determination in predicting $D$ using $h(Z,X)$.

This expression for the asymptotic variance connects efficient IV estimation to optimal prediction. \autoref{eqn:predictionanalogy} shows that the asymptotic variance $V^h$ is strictly increasing in first stage mean-squared error. Moreover, $V^h$ is always weakly greater than the variance that would be obtained by regressing the outcome $Y$ directly on treatment $D$, a generally infeasible benchmark that could only be achieved if $D=h(Z,X)$ a.s. This scenario would imply full compliance with the instrument $h(Z,X)$, which of course would obviate the need for IV estimation in the first place.

The correspondence between low prediction loss and low variance of the implied IV estimator in \autoref{eqn:predictionanalogy} has motivated the use of powerful machine-learning prediction algorithms to estimate instruments of the form $h(Z,X) = f(Z)$, where $Z$ is often assumed to be high-dimensional, for example using the LASSO \citep{Belloni2012-nu} or ridge regression \citep{Hansen2014-wy}.
This literature has primarily focused on minimizing $V^h$ by solving the standard prediction problem of minimizing $\E[(D - h(Z,X))^2]$ through transformations of complex (or many) instruments $Z$. In applied research, where valid instruments are often simple and scarce while covariates are comparably rich and plentiful, these results can still be applied by considering instruments obtained by interaction with the covariates. Here, we explicitly consider such interacted instruments $h(Z,X)$ for which $Z$ is simple, but $X$ may be high-dimensional, and exploit the specific structure of $h(Z,X)$ in that case.

When we consider augmented instruments of the form $h(Z,X)$ to minimize the variance $V^h$ in \autoref{eqn:predictionanalogy}, we face an additional challenge:
we have to ensure that the instrument remains valid -- that is,
we want to solve (in the population distribution)
\begin{align*}
    &\min \E[(D - h(Z,X))^2]
    &
    &\text{s.t. } \Cov(\varepsilon, h(Z,X)) = 0.
\end{align*}
When is the constraint fulfilled? For a simple real-valued instrument $Z$ with $\E[Z] = p$, setting $h(Z,X) = w(X) \: (Z-p)$ for some conformal $w(\cdot)$ is sufficient.
Indeed, if $Z$ is binary, then (up to trivial transformation) any instrument that ensures orthogonality by construction has this form:
\begin{proposition}[Characterization of augmented extensions of a binary instrument]
    \label{prop:binarycharacterization}
    Assume that $Z$ is binary with $\E[Z]=p \in (0,1)$.
    Consider some conformal function $h(\cdot,\cdot)$.
    Then $\Cov(\varepsilon,h(Z,X)) = 0$ for all distributions of $(X,\varepsilon)$ with $\Var(h(Z,X)), \Var(\varepsilon) < \infty$ and $(X,\varepsilon) \orth Z$ if and only if $h(z,x) = \psi + w(x) \: (z-p)$ for some conformal function $w(\cdot)$ and constant $\psi$.
\end{proposition}
When $Z$ is not binary, there are additional instruments; we discuss this case along with a generalization of \autoref{prop:binarycharacterization} to finitely supported instruments in \autoref{sec:generallinear}, yielding a linear combination of augmented instruments of the above form.

For the binary case, we can now solve for the augmented instrument of the form $h(Z,X) = w(X) \: (Z-p)$ that minimizes the prediction loss
\begin{align*}
    \E[(D - h(Z,X))^2]
    =
    \E[(D - w(X) \: (Z-p))^2],
\end{align*}
maintaining the assumption that $\E[\varepsilon^2|X] \equiv \sigma^2$ a.s. 
The optimal interaction term takes a particularly simple form:
\begin{proposition}[Optimal augmented instrument]
    \label{prop:optimalweight}
    Among augmented extensions of a simple binary instrument that are valid in the sense of \autoref{prop:binarycharacterization},
    \begin{align*}
        h(Z,X) = \underbrace{\left(\E[D|Z=1,X] - \E[D|Z=0,X]\right)}_{=\alpha(X)} \: (Z - p)
    \end{align*}
    yields minimal asymptotic variance $V^h$ in \autoref{eqn:predictionanalogy}.
\end{proposition}

The optimal interaction term
\begin{align*}
    \alpha(x) = \E[D|Z=1,X=x] - \E[D|Z=0,X=x]
\end{align*}
is the local strength of the instrument (that is, the net causal effect of the instrument on treatment). This result directly mirrors \cite{abadie2019instrumental} and \cite{Huntington-Klein} who consider compliance that varies by groups, in which case a group-interacted estimator yields compliance weighting. Any other choice for $w(X)$ leads to a prediction loss of
\begin{align*}
    \E[(D - w(X) \: (Z-p))^2]
    = 
    p(1-p) \E\left[(w(X) - \alpha(X))^2 \right]
    +
    \text{const.}
\end{align*}
in \autoref{eqn:predictionanalogy}, so solving for a low-variance estimator corresponds to the problem of estimating the treatment effect of $Z$ on $D$ with low mean-squared error.
Recognizing that the minimization of IV variance is in fact a question of (first-stage) treatment effect estimation leads us to straightforwardly adapt machine learning tools previously developed to estimate heterogeneous treatment effects, which we discuss in \autoref{sec:estimation} and \autoref{sec:simulation}.

\subsection{Connection of interacted to weighted instruments}
\label{sec:weighted}

Assuming that $w(X) \geq 0$ with $\E[w(X)] > 0$, we can shed additional light on the resulting estimator by noting that it can be understood as the weighted estimator
\begin{align*}
    \hat{\tau}^h = \frac{\sE[ w(X) (Y - \sE[Y]) (Z- p)]}{\sE[w(X)  (D - \sE[D]) (Z- p)]} = \frac{\sE^w[ (Y - \sE[Y]) (Z- p)]}{\sE^w[ (D - \sE[D]) (Z- p)]}
\end{align*}
where $\sP^w$ denotes empirical measure with weights $w(X_i)$ (so $\sE^w[\circ] = \frac{\sE[w(X) \: \circ]}{\sE[w(X)]}$, etc.).
Alternatively, we can consider the fully weighted IV estimator
\begin{align*}
    \hat{\tau}_w
    =
    \frac{\sCov^w(Y,Z)}{\sCov^w(Y,D)} = \frac{\sE^w[ (Y - \sE^w[Y]) (Z- p)]}{\sE^w[ (D - \sE^w[D]) (Z- p)]}
\end{align*}
that embeds weighting within the empirical covariance. While $\hat{\tau}^h$ with $h(Z,X) = w(X) \: (Z-p)$ and $\hat{\tau}_w$ differ slightly in how they handle baseline variation, once we control for the weight $w(X)$ in the IV regression, they are equivalent:

\begin{proposition}
    [Equivalence to weighting estimator]
    \label{prop:controlequivalence}

    Consider the IV estimator $\hat{\tau}^h$ with instrument $h(Z,X) = w(X) \: (Z-p)$ and the weighted IV estimator $\hat{\tau}_w$ with weights $w(X)$ and instrument $Z$.
    If we residualize outcomes $Y$ and treatment $D$ linearly with respect to a constant and the weight $w(X)$ (and possibly additional control variables),
    then both are equivalent,
    $
        \hat{\tau}^h = \hat{\tau}_w.
    $
\end{proposition}

These two results show the equivalence of certain interacted instruments \citep[for interacted instruments see e.g.][]{Bond2007-fd} to weighted instruments \citep{Joffe2003-uq}, which are also considered in \citet{Huntington-Klein}. This connection allows us to obtain efficiency in specific cases. For example, when $D$ is binary and there are no defiers, then $\alpha(X)$ is the conditional probability of compliance. (We will make this language precise below using a standard potential outcomes framework.) In this setting, \citet{Joffe2003-uq} present conditions under which compliance weighting is efficient in the homoscedastic and homogeneous case.
For the remainder of this paper, we will impose the assumptions of this proposition, so we write $\hat{\tau}_w$ for the equivalent estimators, and call it the ``$w$-weighted IV estimator''. 

\subsection{Interpretation in potential outcomes framework and the weighted LATE}
\label{sec:generalbinary}

In this section, we consider the case of a binary instrument $Z$ and binary treatment $D$ in a potential outcomes framework, dropping the assumptions of homoscedasticity and treatment effect homogeneity.
We adopt the standard assumptions for identification of the local average treatment effect (LATE) for the population distribution, also including covariates $X$ that are not affected by the instrument:
\begin{enumerate}
    \item \underline{Exclusion}: The potential outcomes only depend on the instrument $Z$ through treatment $D$, so we can write $Y(1), Y(0)$ for the potential treatment and control outcomes, where $Y= Y(D)$ is observed.
    \item \underline{Random assignment of the instrument}:
    \begin{align*}
        (Y(1),Y(0),D(1),D(0),X) \orth Z,
    \end{align*}
    with $D(1),D(0)$ representing the potential treatment outcomes as a function of assignment, where $D=D(Z)$ is observed.
    \item \underline{Relevance}: For $p= \E[Z] = \P(Z=1)$,
    \begin{align*}
        &0 < p < 1,
        &
        \P(D(1) = 1) &\neq \P(D(0) = 1).
    \end{align*}
    \item \underline{Monotonicity}:
    \begin{align*}
        D(1) \geq D(0).
    \end{align*}
\end{enumerate}
The last assumption implies that there are no defiers (units with $D(1)=0 < 1 = D(0)$), which allows us to identify conditional local average treatment effects
\begin{align*}
    \tau(x) = \E[Y(1) - Y(0) | D(1) > D(0), X=x] = \frac{\E[Y| Z=1, X=x] - \E[Y| Z=0, X=x]}{\E[D| Z=1, X=x] - \E[D| Z=0, X=x]}.
\end{align*}
Monotonicity is not essential to these results, although it simplifies interpretation. For instance, if this assumption were violated, we could still consider estimating $\tau(x) = \frac{\E[Y| Z=1, X=x] - \E[Y| Z=0, X=x]}{\E[Y| D=1, X=x] - \E[Y| D=0, X=x]}$ for those values of $x$ with $\E[D| Z=1, X=x] > \E[D| Z=0, X=x]$. However, the parameter loses its causal interpretation if treatment effects differ for compliers and defiers with $X=x$ \citep{angrist1996}.

Given an iid sample of size $n$ from the population distribution of $(Y,D,Z,X)$,
the standard Wald estimator of the local average treatment effect $\tau = \E[Y(1) - Y(0) | D(1) > D(0)]$ is the IV estimator
\begin{align*}
    \hat{\tau} = \frac{\sE[Y| Z=1] - \sE[Y| Z=0]}{\sE[D| Z=1] - \sE[D| Z=0]} = \frac{\sCov(Y,Z)}{\sCov(D,Z)},
\end{align*}
where we write $\sE,\sCov$ for the expectation and covariance operators associated with the empirical (sample) measure $\sP$ (So $\sE[\circ] = \frac{1}{n} \sum_{i=1}^n \circ_i$, etc.).
This is also the IV estimator in the model
\begin{align*}
    Y &= \psi + D \tau + \varepsilon,
    &
    \E[\varepsilon] &= 0,
\end{align*}
where $\Cov(Z,\varepsilon) = 0$.
(Throughout, we assume the existence of sufficient moments to allow for the consistent estimation of $\tau$ in this manner.)

When conditional treatment effects $\tau(X)$ vary, then the $w$-weighted IV estimator $\hat{\tau}_w$ is not necessarily consistent for the overall LATE $\tau = \E[Y(1) - Y(0) |D(1) > D(0)] = \E[\tau(X)]$. Instead, the estimand corresponding to this estimator now represents a weighted average of treatment effects among compliers:

\begin{proposition}
    [Weighted LATE]
    \label{prop:weightedlate}

    Under standard regularity conditions, the $w$-weighted IV estimator $\hat{\tau}_w$ is consistent for the $w$-weighted LATE $\tau_w$,
    \begin{align*}
        \hat{\tau}_w \cp \tau_w = \frac{\E[w(X) \: \tau(X) |D(1) > D(0)]}{\E[w(X)|D(1) > D(0)]}
    \end{align*}
    as $n \rightarrow \infty$.
\end{proposition}

In the previous section we followed \cite{Joffe2003-uq} in arguing in favor of weighting by the compliance probability
\begin{align*}
    \alpha(X=x) = \E[D|Z=1,X=x] - \E[D|Z=0,X=x] = \P(D(1) > D(0)|X=x).
\end{align*}
However, the LATE is already a compliance-weighted average of conditional LATEs \citep{angrist1995two}, which we write as 
\begin{align*}
    \tau = \E[Y(1) - Y(0) |D(1) > D(0)]
    =  \frac{\E[\Ind_{D(1) > D(0)} \: \tau(X) ]}{\E[\Ind_{D(1) > D(0)}]}
    = \frac{\E[\alpha(X) \: \tau(X)]}{\E[\alpha(X)]}.
\end{align*}
The compliance-weighted IV estimator $\hat{\tau}^* = \hat{\tau}_{\alpha}$ reduces the contribution of always- and never-takers to the variance by ``doubling down'' on that weighting, estimating the compliance-weighted LATE 
\begin{align*}
    \tau_{\alpha} = 
    \frac{\E[\alpha(X)\: \tau(X)|D(1) > D(0)]}{\E[\alpha(X)|D(1) > D(0)]}
    =\frac{\E[(\alpha(X))^2 \: \tau(X)]}{\E[(\alpha(X))^2]},
\end{align*}
corresponding to a ``Super-Local Average Treatment Effect'' in the linear model considered by \citet{Huntington-Klein}.
Furthermore, heterogeneity and heteroscedasticity can also affect the asymptotic distribution around the estimand:

\begin{proposition}[Asymptotic distribution of $w$-weighted IV estimator]
    \label{prop:asymptotic}

    Under regularity conditions,
    \begin{align*}
        \sqrt{n}
        (\hat{\tau}_w - \tau_w)
        &\cd
        \N(0, V_w),
        &
        V_w &= \frac{\E\left[w^2(X) \: \sigma^2(X) \right]}{p(1-p)\E^2[\alpha(X) \: w(X)]}
    \end{align*}
    as $n \rightarrow \infty$
    for
    $
        \sigma^2(x) = (1-p) \E[\varepsilon^2|Z=1,X=x] + p \E[\varepsilon^2|Z=0,X=x]
    $
    where
    $\varepsilon$ is the population regression residual in the linear regression of $Y - \tau_w D$ on a constant, $w(X)$, and any additional control variables included in the IV regression.
\end{proposition}

The conditional residual variation $\sigma^2(x) = (1-p) \E[\varepsilon^2|Z=1,X=x] + p \E[\varepsilon^2|Z=0,X=x]$ comes from three sources: from variation of conditional LATEs $\tau(x)$ around their weighted average $\tau_w$; from variation in the outcome that can be predicted from pre-assignment characteristics, which can be addressed by residualizing the outcome with respect to (potentially nonparametric) functions of $X$; and from irreducible noise that cannot be explained by $X$.

While we leave a more thorough treatment of these additional sources of variation (and their mitigation) to future versions,
we note here that the Cauchy--Schwartz inequality immediately implies that the variance-minimizing weights
are given by
\begin{align*}
    w^*(x) &= \frac{\alpha(x)}{\sigma^2(x)}
    &
    &\text{yielding asymptotic variance}
    &
    V_{w^*} &= \left(p (1-p) \E\left[\frac{(\alpha(X))^2}{\sigma^2(X)}\right] \right)^{-1}
\end{align*}
and a natural extension to the estimators in this paper is to also estimate $\sigma^2(X)$ when determining weights \citep[e.g., as proposed by][for linear regression using machine learning]{Miller2018-xt}, in addition to compliance $\alpha(X)$. That said, an estimator that weights by compliance alone may prove to be desirable in many empirical settings due to concerns regarding the finite-sample performance of inverse-variance estimators, analogous to the reason that an applied researcher may prefer OLS to the asymptotically efficient GLS.

We also note a parallel to \citet{Borusyak2020-xu}, in which the authors develop methods to improve IV efficiency by recentering the instrument to account for heterogeneous treatment exposure. In a more general setting where the given instrument is not valid unconditionally, they extend the framework of \citet{chamberlain1987asymptotic}, and propose an optimal instrument that can be re-written in our notation and under our assumptions as
\begin{align*}
    Z^*
    =
    \frac{1}{\sigma^2(X)}
    \left(\E[D|Z,X] - \E[D|X]\right)
    =
    \frac{\alpha(X)}{\sigma^2(X)} (Z-p),
\end{align*}
corresponding to the efficient weighting function $w^*(x)$.
Motivated by the homoscedastic and homogeneous case where $\sigma^2(X) \equiv \sigma^2$, we focus for now on the weight choice $w(x) = \alpha(x)$, which is optimal in this case. This has the additional benefit of not requiring any information about the distribution of outcome $Y$, while also giving the weights in $\tau_\alpha$ an immediately intuitive meaning.

How should $\hat{\tau}_\alpha$ be interpreted in practice? Under assumptions of treatment effect homogeneity and homoscedasticity, it is a consistent estimator of the LATE, and more efficient than traditional IV estimators (e.g., 2SLS). Moreover, consistency for the LATE also extends to settings in which treatment effects are heterogeneous but uncorrelated with the compliance weights. For example, consider an RCT designed to estimate the effectiveness of a medical intervention administered at an outpatient health clinic. Suppose that the distance between study participants' homes and the clinic is highly predictive of their compliance with random assignment, but that the effectiveness of the treatment is largely determined by the participants' latent genetic characteristics. In such settings, it seems unlikely that estimated compliance would be meaningfully correlated with treatment effects, in which case compliance-weighted estimates could be simply interpreted as approximations to the LATE.

However, the estimand $\tau_w$ and the LATE $\tau$ are not generally equivalent when $\tau(x)$ are both heterogeneous and correlated with $w(x)$; in this setting, the estimand $\tau_\alpha$ represents a convex combination of conditional LATEs \cite[``Super-LATE'' in the parlance of][]{Huntington-Klein}. Yet even when the unconditional LATE is the estimand of interest, we point out that in many applied scenarios, the compliance-weighted IV estimator $\hat{\tau}_\alpha$ will nonetheless be more precise for the LATE than the traditional IV estimator $\hat{\tau}$.
To do so, we consider an asymptotic framework with relatively small treatment effects. When treatment effects are fixed in our asymptotic framework, then we can distinguish the LATE arbitrarily well from any alternative in a large sample. So instead, we focus on the case where treatment effects are of an order of magnitude that makes them hard to distinguish from noise, a common scenario in applied research.%
\footnote{Note that this approach is similar to the weak instrument analysis of \citet{staiger1997instrumental} in that we adopt a local-to-zero regime, but complementary in that rather than applying it to first-stage effects, we apply it to treatment effects, a separate and generally less problematic source of weak signal.}

\begin{proposition}[Asymptotic error in estimating LATE using $w$-weighted IV estimator]
    \label{prop:asymptoticlocal}

    Assume that conditional local average treatment effects are local to zero in the sense that
    $
        \tau(x) = \frac{\mu(x)}{\sqrt{n}}
    $.
    Then
    \begin{align*}
        \sqrt{n} (\hat{\tau}_w - \tau)
        &\cd
        \N\left(
            B_w,
            V_w
        \right)
        &
        B_w
        &=
        \frac{\Cov(\mu(X), w(X)|D(1)>D(0))}{\E[w(X)|D(1)>D(0)]}
    \end{align*}
    as $n \rightarrow \infty$, and $V_w$ does not depend on treatment effects.
    
\end{proposition}

Hence, if the unconditional LATE $\tau$ is in fact the estimand of interest, the $w$-weighted estimator has an asymptotic bias $B_w$ that is equal to the covariance of conditional LATEs and relative weights $w(X)/\E[w(X)]$ among only the compliers. This result is intuitive: since only compliers contribute to the conditional LATEs, the extent to which $\tau_w$ can deviate from $\tau$ is governed by differences in the relative weighting of compliers. 
This also highlights the fact that when weighting by compliance (i.e., $w(x) = \alpha(x)$), $B_\alpha = 0$ both when $X$ is not predictive of compliance (in which case $w(X) \equiv \alpha$), as well as when $X$ is perfectly predictive of compliance (in which case $w(x) = 1$ for all compliers, and $w(x) = 0$ for all non-compliers).

In principle, the fit from estimating the unweighted LATE $\tau$ by the $w$-weighted IV estimator $\hat{\tau}_w$ can be asymptotically better or worse than using the standard unweighted IV estimator $\hat{\tau}^{\text{IV}}$.
Indeed, the expected square of the asymptotic approximation to $\sqrt{n} (\hat{\tau}_w - \tau)$ is
\begin{align}
    \label{eqn:asymptloss}
    \frac{\Cov^2(\mu(X), w(X)|D(1)>D(0)) + \frac{\E[w^2(X)\sigma^2(X)]}{p(1-p) \P^2(D(1)>D(0))}}{\E^2[w(X)|D(1)>D(0)]}
\end{align}
representing a bias--variance trade-off that can make it smaller or larger than the variance of the large-sample limit of $\sqrt{n} (\hat{\tau}^{\text{IV}} - \tau)$, which is $\E\left[\sigma^2(X)\right] / (p(1-p) \P^2(D(1)>D(0)))$.

A particularly simple estimator that navigates this bias--variance trade-off is the estimator $\hat{\tau}^*(\lambda)$ that implements weights
\begin{align*}
    w(x) = (1-\lambda) \: \frac{\E[(\alpha(x))^2]}{\E[\alpha(x)]}  + \lambda \: \alpha(x) = (1-\lambda) \: \E[\alpha(x)|D(1)>D(0)]  + \lambda \: \alpha(x)
\end{align*}
for $\lambda \in [0,1]$, interpolating between equal weighting ($\lambda=0$) and compliance weights ($\lambda=1$).
This weighting scheme can be interpreted as a shrinkage estimator in the population compliance prediction problem.
When compliance probabilities do not co-vary too much with treatment effects and conditional variances, then the resulting estimator has better fit for the LATE than the standard IV estimator:

\begin{corollary}
    [Dominance of shrinkage LATE estimator]
    \label{cor:dominance}
    If
    \begin{align*}
        \Cov\left(\frac{\alpha(X)}{\E[\alpha(X)]},\sigma^2(X)\right) < \E[\sigma^2(X)] \: \Var\left(\frac{\alpha(X)}{\E[\alpha(X)]}\right)
    \end{align*}
    then for $\lambda$ sufficiently small, the sum of variance and squared bias of the asymptotic distribution of $\hat{\tau}^*(\lambda)$ around the LATE $\tau$ is less than the variance of the standard (unweighted) IV estimator $\hat{\tau}^\text{IV}$.
\end{corollary}
\noindent This condition holds, in particular, under homoscedasticity or when the variation of conditional variances is small.

As an alternative to this analysis of the estimator relative to the LATE, compliance-weighted estimators could be analyzed within the marginal treatment effect (MTE) frameworks \citep{heckman2005structural,heckman2006understanding,brinch2017beyond,mogstad2018using}, which use IV to bound (or extrapolate to) treatment effect parameters of interest that cannot be point-identified, such as the ATE. In applications of these frameworks \citep[e.g.][]{kowalski2016doing,kowalski2018examine,kowalski2018reconciling}, compliance-weighted estimators may similarly improve precision in estimating treatment-effect bounds.

    \section{Estimation and Inference with Estimated Weights}
\label{sec:estimation}

In this section we analyze the asymptotic distribution of the weight-augmented IV estimator when weights are estimated, rather than fixed ex-ante, and consider inference. Under regularity conditions, we find that weight estimation does not change the asymptotic distribution relative to a specific sample-dependent estimand, and that commonly used robust standard errors remain valid after first-stage estimation. Notably, our assumptions allow for valid inference relative to this data-dependent estimand even when compliance weights are estimated nonparametrically.

In the previous section, we considered IV estimation with instruments of the form $h(Z,X)$ for a simple instrument $Z$ that is independent of covariates $X$, focusing on cases where $h(Z,X) = w(X) \: (Z-p)$ for $p = \E[Z]$. In the homogeneous and homoscedastic baseline case with a binary instrument $Z$, we justified the specific choice
\begin{align*}
    \alpha(x) = \E[D|Z=1,X=x] - \E[D|Z=0,X=x]  
\end{align*}
for the weight $w(X)$, corresponding to the probability of compliance in an instrumental-variable setting with binary treatment and no defiers.
Yet in practice, such weights are unknown to the researcher, and must be estimated instead.
We therefore analyze the asymptotic distributions of the augmented IV estimator $\hat{\tau}^{\hat{h}}$ and the weighted IV estimator $\hat{\tau}_{\hat{w}}$, where
\begin{align*}
    \hat{\tau}^{\hat{h}} &= \frac{\sCov(Y, (Z - \sE[Z]) \hat{w}(X))}{\sCov(D, (Z - \sE[Z]) \hat{w}(X))},
    &
    \hat{\tau}_{\hat{w}} &= \frac{\sCov^{\hat{w}}(Y,Z)}{\sCov^{\hat{w}}(D, Z)}
\end{align*}
for the empirical covariance $\sCov$ and weighted empirical covariance $\sCov^{\hat{w}}$.

\begin{proposition}[Consistency with estimated weights]
    \label{prop:estimatedconsistent}
    Assume that weights $\hat{w}_{-j(i)}(X_i)$ are fitted with $k$-fold cross-fitting,
    where folds are (approximately) equally sized with $X_i$ in fold $j=j(i)$ and $\hat{w}_{-j}$ fitted on all other folds.
    Also assume that there is some function $w(\cdot) \geq 0, \E[w(X)] > 0$ such that $\E[(\hat{w}_{-j}(X) - w(X))^2] \rightarrow 0$ as $n \rightarrow \infty$ (where the expectation is both over the training data and an independently drawn test point) for all $j$.
    Then $\hat{\tau}^{\hat{h}}, \hat{\tau}_{\hat{w}} \cp \tau_w$ under regularity conditions.
\end{proposition}

We note that the notion of convergence of $w_{-j}$ to $w$ as well as the cross-fitting approach we employ for weights are similar to the assumption and construction \citet{Chen2020-xg} employ for instruments with a linear second-stage model, while we focus specifically on compliance estimation in a completely nonparametric model.

While \autoref{prop:estimatedconsistent} establishes consistency, we also care about the asymptotic distribution.
$\hat{\tau}^{\hat{h}}, \hat{\tau}_{\hat{w}}$ can generally differ and have different asymptotic distributions.
We now residualize $Y_i$ and $D_i$ linearly with respect to a constant and the weight $\hat{w}_{-j(i)}(X_i)$ (and potentially additional pre-assignment control variables) and invoke \autoref{prop:controlequivalence}, which applies to the estimators with estimated weights as well, to conclude that $\hat{\tau}^{\hat{h}} = \hat{\tau}_{\hat{w}}$ in that case.
Assuming this residualization throughout, we can thus focus on $\hat{\tau}_{\hat{w}}$.

Under arbitrary treatment effect heterogeneity and estimated weights, there is little hope of establishing $\sqrt{n}$-consistency relative to the $w$-weighted LATE $\tau_w$ without imposing stronger assumptions on the estimation of weights and/or the heterogeneity of conditional LATEs.
Instead of imposing such assumptions ex-ante, we consider an estimand that is itself a function of the data, namely the cross-fitted hybrid estimand
\begin{align*}
    \tau_{\hat{w}}^{\text{CF}}
    =
    \frac{\frac{1}{k} \sum_{j=1}^k \E[\alpha(X) \: \hat{w}_{-j}(X) \: \tau(X) |\hat{w}_{-j}]}{\frac{1}{k} \sum_{j=1}^k \E[\alpha(X) \: \hat{w}_{-j}(X)  |\hat{w}_{-j}]},
\end{align*}
where conditional expectations consider a draw of the data independent of $\hat{w}_{-j}$.
In other words, this estimand represents a $\hat{w}^{\text{CF}}(\cdot) = \frac{1}{k} \sum_{j=1}^k \hat{w}_{-j}(\cdot)$-weighted LATE that depends itself on the data through the cross-fitting splits and weight functions.
Relative to this estimand, we obtain an asymptotic distribution that is equivalent to that for fixed weights:

\begin{proposition}[Asymptotic distribution with estimated weights]
    \label{prop:estimatedasymptotic}
    Under the assumptions of \autoref{prop:estimatedconsistent} and regularity conditions,
    \begin{align*}
        \sqrt{n} (\hat{\tau}_{\hat{w}} - \tau_{\hat{w}}^{\text{CF}})
        & \cd \N(0,V_w),
        &
        V_w &= \frac{\E\left[w^2(X) \: \sigma^2(X) \right]}{p(1-p)\E^2[\alpha(X) \: w(X)]}
    \end{align*}
    as $n \rightarrow \infty$,
    where
    $
        \sigma^2(x) = (1-p) \E[\varepsilon^2|Z=1,X=x] + p \E[\varepsilon^2|Z=0,X=x],
    $
    and $\varepsilon$ is the population regression residual in the linear regression of $Y - \tau_w D$ on a constant, the weight $w(X)$, and any other control variables (where we assume that $Y_i$ and $D_i$ have been residualized accordingly).
\end{proposition}

This result is related to the idea in \citet{Chernozhukov2017-ug} of separating out variation conditional on sample splits from variation through sample splits themselves. We simply focus on the former and take the sample splits -- and corresponding weights -- as given, while ignoring variation through sample splits and weight estimation itself. The ``variational'' inference approach from \citet{Chernozhukov2017-ug} could complement this purely conditional inference.
As an alternative, we could have considered the empirically-weighted estimand
\begin{align*}
    \tau^{\text{empirical}}_{\hat{w}} = \frac{\sE[\alpha(X) \: \hat{w}(X) \: \tau(X)]}{\sE[\alpha(X) \: \hat{w}(X)]}
    =
    \frac{\sum_{i=1}^n \alpha(X_i) \: \hat{w}_{-j(i)}(X_i) \: \tau(X_i)}{\sum_{i=1}^n \alpha(X_i) \: \hat{w}_{-j(i)}(X_i)}
\end{align*}
reducing the asymptotic variance further to the residual in a regression of $Y - \tau(X) D$, but ignoring additional within-fold variation.

When treatment effects have variation that vanishes sufficiently quickly -- for example in the case where treatment effects are themselves local to zero -- we obtain $\sqrt{n}$-consistency relative to the $w$-weighted LATE:
\begin{corollary}[$\sqrt{n}$-consistency]
    Assume also that $\Var(\tau(X)) = \Oh(1/n)$ (for example, $\tau(x) = \frac{\mu(x)}{\sqrt{n}}$ as in \autoref{prop:asymptoticlocal}).
    Then
    $
        \sqrt{n} (\hat{\tau}_{\hat{w}} - \tau_w)
        \cd \N(0,V_w)
    $
    as $n \rightarrow \infty$.
\end{corollary}

When we aim to obtain compliance weights $\alpha(x) = \E[D|Z=1, X=x] - \E[D|Z=0, X=x]$ for the sake of minimizing the empirical variance $V_w$ in the homogeneous and homoscedastic case, we face a causal estimation problem, since the compliance score $\alpha(X)$ is not directly observed -- for each observation $i$, we only observe treatment $D=D_i$ given the realization $Z=Z_i$, but not for the counterfactual assignment $Z=1 - Z_i$.
We therefore face a first stage heterogeneous treatment effect estimation problem (where the instrument $Z$ represents treatment, and the treatment variable $D$ is the outcome of interest).
Two natural estimation strategies in this case are:
\begin{enumerate}
    \item Using a low-dimensional fully interacted model to estimate compliance, such as linear (or logistic) regression, and obtaining weights from the positive part of the in-sample fit. In that case, the target weight function $w$ is (the positive part of) a model-specific approximation of $\alpha$, such as (the positive part of) the best linear predictor.
    \item Using a flexible nonparametric tool, such as the causal forest \citep{wager2018estimation}. In that case, the target weight function is $\alpha$, and the approximation is the (positive part of the) out-of-fold estimate. 
    This is similar to \citet{Huntington-Klein}, who also uses the causal forest for compliance estimation, but instead groups observations by quintiles of honest compliance estimates.
\end{enumerate}
In general, with a binary instrument, any algorithm for the consistent estimation of heterogeneous treatment effects from experimental data is also a tool for the estimation of compliance weights in our setting, and better estimation translates into lower asymptotic variance under homogeneity and homoscedasticity.
Additional shrinkage that goes beyond the regularization employed to reduce out-of-sample mean-squared error in prediction may be helpful to mitigate the impact of heterogeneous treatment effects when the local average treatment effect $\tau$ is the estimand of interest.

Finally, for inference on treatment effects, we also want to estimate the asymptotic variance
\begin{align*}
    V_w &= \frac{\E\left[w^2(X) \: \sigma^2(X) \right]}{p(1-p)\E^2[\alpha(X) \: w(X)]},
    & 
    \sigma^2(x)
    &=
    (1-p) \E[\varepsilon^2|Z=1,X=x] + p \E[\varepsilon^2|Z=0,X=x]
\end{align*}
where $\varepsilon$ are the residuals in a linear regression of $Y - D \tau_w$ on the control variables, including the weight $w(X)$.
We can consistently estimate the asymptotic variance by the sample analogue (where we assume that the instrument is binary):

\begin{proposition}
    [Valid variance estimation]
    \label{prop:robustse}
    Under the assumptions of \autoref{prop:estimatedasymptotic} and regularity conditions,
    \begin{align*}
        \hat{V}_{\hat{w}}
        =
        \frac{ \left(\frac{\sE\left[\hat{w}^2(X) \: \hat{\varepsilon}^2|Z=1\right]}{\sP(Z=1) \sE[\hat{w}^2(X)]} + \frac{\sE\left[\hat{w}^2(X) \: \hat{\varepsilon}^2|Z=0\right]}{\sP(Z=0) \sE[\hat{w}^2(X)]}\right) \sVar^2(Z) }{\left(\sCov^{\hat{w}}(Z , D)\right)^2}
        \cd V_w
    \end{align*}
    for residuals $\hat{\varepsilon}_i$ in the linear regression of $Y_i - \hat{\tau}_{\hat{w}} D_i$ on the control variables, including $\hat{w}(X_i)$.
\end{proposition}

Standard heteroscedasticity-robust variance estimators in the weighted IV regression (after residualization) with estimated weights similarly remain valid, where $\P(Z{=}1), \P(Z{=}0)$, and $\Var(Z)$ are typically estimated with weights (which are asymptotically equivalent when $X$ and $Z$ are independent, as in our setting).

    \section{Simulation}
\label{sec:simulation}

This section presents a simulation study that systematically compares the performance of the compliance-weighted IV estimator under a number different weight estimation procedures to that of unweighted IV. We consider a number of data generating processes that vary the degree of heterogeneity in compliance and treatment effects, as well as the heteroscedasticity of errors.

\subsection{Data-generating processes}
\label{sec:sim_DGP}

In this section, we demonstrate the finite-sample performance of the estimation procedures discussed in \autoref{sec:estimation}. Let $\delta_i, \varepsilon_i, \tau_i, \sim \mathcal{N}(\mathbf{0}, \Sigma)$, with  
\begin{align*}
	\Sigma &= 
	\begin{pmatrix} 
		1 & \rho_{\delta\varepsilon} & \rho_{\delta\tau}\sigma_{\tau} \\ 
		\rho_{\delta\varepsilon} & 1  & \rho_{\tau\varepsilon}\sigma_{\tau}\\ \rho_{\delta\tau}\sigma_{\tau}  & \rho_{\tau\varepsilon}\sigma_{\tau} & \sigma_{\tau}^2
	\end{pmatrix}
\end{align*}
where $\delta_i$ is the latent tendency to receive treatment, $\tau_i$ is the treatment effect, and $\varepsilon_i$ is the baseline untreated potential outcome for individual $i$.

Using the framework outlined in \autoref{sec:generalbinary}, let potential treatment indicators be defined as
\begin{align*}
	D_i(0) &= \Ind(\Phi^{-1}(\delta_i) > 1 - S_{AT}),
	&
	D_i(1) &= \Ind(\Phi^{-1}(\delta_i) > S_{NT})
\end{align*}
where $S_{AT}$ and $S_{NT}$ represent the share of always-takers and never-takers in the population, respectively. Realized outcomes and treatments are thus defined as
\begin{align*}
Y_{i} &= D_i\tau_{i} + (1+\zeta\delta_i)\err_{i},
&
D_i &= D_i(0)(1-Z_i) + D_i(1)Z_i
\end{align*}
in which scalar parameter $\zeta$ determines the degree of heteroscedasticity exhibited in the variance of the untreated outcome. When $\rho_{\delta\varepsilon} \neq 0$, observations in this simulation exhibit selection into treatment that is correlated with untreated potential outcomes, a setting in which IV methods can prove useful. Moreover, when $\rho_{\delta\tau} \neq 0$,
there will be differences between the unweighted and $w$-weighted IV estimands, as discussed in \autoref{sec:generalbinary}.

In order to explore how the performance of the compliance-weighted estimator relates to the researcher's ability to predict compliance, we employ a framework that allows us to vary the irreducible error in predicting treatment. Let 
\begin{align*}
X_i = \delta_i + \eta_i
\end{align*}
represent a noisy continuous predictor of treatment receipt, where $\eta_i \sim \mathcal{N}(0, \sigma_{\eta}^2)$ is a random error. Note that while $X$ is a linear function of $\delta$, compliance rate $\alpha(X) = \E[D|Z=1] - \E[D|Z=0]$ is not, making estimation a non-trivial exercise. We demonstrate the importance of selecting appropriate compliance estimation techniques for improving IV estimate precision and conducting valid inference in \autoref{sec:sim_results} below.

Using this framework, we produce simulations under four different data generating processes (DGPs). In DGP 1, the baseline case, there is a constant treatment effect and homoscedasticity. In DGP 2, there are heterogeneous treatment effects that are positively correlated with the latent tendency to receive treatment ($\rho_{\delta\tau} = 0.5$), consistent with a Roy-type model of selection on heterogeneous treatment effects. In DGPs 3 and 4, errors are heteroscedastic and positively ($\zeta = 0.25$) and negatively ($\zeta = -0.25$) correlated with covariates, respectively. For each DGP, we draw $M = 1,000$ Monte Carlo samples of $N = 1,000$ observations and generate treatment effect estimates using both unweighted IV and the compliance-weighted IV estimators while varying irreducible error $\sigma_{\eta} \in \{0.5, 1, 2\}$. In each simulation, always-takers, compliers, and never-takers comprise $0.05$, $0.25$, and $0.70$ of the population, respectively.

\subsection{Compliance weight estimation methods}
\label{sec:sim_methods}

To demonstrate the role of compliance estimation in the performance of the weighted estimator, we estimate compliance weights using three different feasible estimation procedures: in-sample OLS, cross-fitted OLS, and cross-fitted causal forest. 
In-sample weights are the compliance estimates generated by fitting a prediction function to the sample from which it was constructed. On the other hand, cross-fitted weights are out-of-sample fits produced using 5 randomly-partitioned folds, as discussed in \autoref{sec:estimation}.

As mentioned above, the compliance score $\alpha(X)$ is a non-linear function of $X$. So when using OLS to generate weights, we flexibly approximate $\alpha(X)$ by splitting $X$ into an ordinal series of bins, then estimating the compliance rate within each bin. Specifically, let $\hat{t}_{jz}(J)$ represent a function that predicts treatment rates conditional on the instrument and a coarse ranking of $X_i$ that assigns individuals into one of $J<<N$ bins of equal size partitioning the distribution of $X_i$
\begin{align*}
	\hat{t}_{jz}(J) &=  \frac{\sum_{i=1}^{N} \Ind\left(Z_i = z \cap  F_n(X_i) \in \left[\frac{j-1}{J}, \frac{j}{J}\right)\right) D_{i}}{\sum_{i=1}^{N} \Ind\left(Z_i = z \cap F_n(X_i) \in \left[\frac{j-1}{J}, \frac{j}{J}\right)\right)}
\end{align*}
where $F_N(\cdot)$ is the empirical distribution function for a sample of size $N$. Predicted rates are thus the estimated mean treatment rates among those in group $z$ and bin $j$, and the complexity of the function is increasing in $J$. The bin-specific compliance weight estimates are therefore
\begin{align*}
	\hat{w_j} = \hat{t}_{j1}(J) - \hat{t}_{j0}(J).
\end{align*}
Note that these $w_j$ are equivalent to the coefficients from an OLS regression of $D$ on $Z$ interacted with dummy indicators for each bin. We can then control the complexity of this compliance prediction function, and therefore its tendency to over-fit in small samples, by varying $J$. 

Alternatively, when using the causal forest algorithm \citep{wager2018estimation} to estimate the weights, compliance is nonparametrically and ``honestly" estimated by predicting $D$ using $Z$ and $X$ directly (rather than the $J$-bin partitioning of $X$ used in OLS estimation), with regularization parameters chosen to minimize out-of-bag RMSE. In this context, ``honest" estimation means that tree splits on $X$ are selected by the algorithm to maximize compliance heterogeneity using only a randomly-drawn subset of the data, while compliance estimates are fitted using only the observations excluded from this subset. This allows for asymptotically unbiased estimation of compliance conditional on $X$.

\subsection{Results}
\label{sec:sim_results}

\autoref{fig:w_dist} presents the analytically-derived oracle (i.e., population) compliance weights $\alpha(X)$ across the cumulative distribution of $X$ for each value of irreducible error $\eta$.\footnote{See \autoref{sec:sim_deriv} for derivations of the oracle compliance weight distribution.} When irreducible error is relatively small, compliance weights are heavily concentrated among compliers, and a substantial share of observations receives no weight whatsoever. However, as the irreducible error increases, the distribution of compliance weights flattens, placing relatively less weight on compliers. Throughout the simulation results presented below, the compliance-weighted IV estimates generated using these oracle weights represent an upper-bound on the performance of the compliance-weighted IV estimator possessing full (population) information of the joint distribution of observed variables $X$, $D$, and $Z$.

\autoref{fig:density_no_ols} presents both unweighted and compliance-weighted IV estimate distributions corresponding to each of the DGPs described above using oracle weights. Similarly, \autoref{fig:rmse} presents the RMSE performance of the compliance-weighted IV estimator relative to unweighted IV in estimating the LATE. First, in order to temporarily abstract away from the role of compliance weight estimation error, consider the case in which oracle weights are used in each of the simulations. The role of irreducible error first-stage prediction problem (reflected by $\sigma_{\eta}$) is immediately evident in each of these figures. Regardless of the data-generating process, when the error is small (e.g., $\sigma_{\eta} = 0.5$), the reduction in RMSE relative to 2SLS is large (roughly 40\% in DGP 1, the baseline homogeneous case). As $\sigma_{\eta}$ grows, the information content in $X$ weakens, reducing prediction quality and leading to more modest improvements in estimator performance over unweighted IV. 

Next, consider the performance of the compliance-weighted estimator using the feasible OLS weight estimation schemes, considering two cases: one in which the complexity of the prediction function is low ($J = 10$), and one in which the complexity of the prediction function is high ($J = 50$). When $J = 10$, there 100 observations (including an average of 25 compliers) per bin, allows for relatively precise estimation of the bin-specific compliance weights. In this case, weights estimated in-sample generally offer better performance than those estimated through cross-fitting, yielding substantially lower variance and minimal bias.

However, when the complexity of the OLS prediction function is high ($J = 50$), there are only 20 observations (including an average of 5 compliers) per bin. In this case, the potential for bias due to over-fitting is much greater. Indeed, the results in \autoref{fig:decomp_mse} demonstrate that when $X$ is an insufficiently strong predictor of compliance (e.g. when $\sigma_{\eta} = 2$) and weights are over-fitted, the in-sample compliance-weighted IV estimator can exhibit substantial bias. This results in coverage below the nominal rate, as shown in \autoref{tab:coverage_late} and \autoref{tab:coverage_wlate}, which present coverage rates for the compliance-weighted IV estimators with respect to the population LATE and compliance-weighted LATEs, respectively. However, consistent with the analytical results in \autoref{sec:estimation}, cross-fitted compliance-weighted IV estimators allow for valid inference regardless of the level of irreducible error or model complexity. Nonetheless, while the bias resulting from over-fitted weight estimates can be eliminated through the use of cross-fitting, this may introduce variance that exceeds the (squared) bias that it eliminates, potentially resulting in greater MSE.

Yet we need not restrict ourselves to using only unpenalized linear models to generate compliance rates. As described in \autoref{sec:oracle}, the asymptotic variance of the IV estimator for DGP 1 depends on the prediction error of the instrument. This suggests that in many settings, superior performance can be obtained by exploiting the bias-variance trade-off in compliance prediction through regularization. The results in \autoref{fig:rmse} and \autoref{fig:decomp_mse} show that compliance-weighted IV estimates using weights generated by a regularized prediction function such as the causal forest have substantially lower variance than those using unpenalized cross-fitted weights.

Note that regardless of the degree of noise in the first stage ($\sigma_{\eta}$) or the weight estimation procedure employed, the relative performance of the $w$-weighted estimator exhibits a similar pattern across the four DGPs. Relative to baseline case DGP 1 (homogeneous and homoscedastic), performance under heteroscedasticity depends on the manner in which the squared residuals are correlated with the compliance weights. When this correlation is positive (DGP 3), then performance declines, since the variance among compliers is above average; when this correlation is negative (DGP 4), then performance improves, since the variance among compliers is below average. Similarly, the RMSE performance dips slightly when treatment effects are heterogeneous and correlated with the weights (DGP 2), since in this case the compliance-weighted estimand differs from the LATE, as discussed in \autoref{sec:generalbinary}. However, the results in DGP 2 strongly suggest that  when the heterogeneity in treatment effects is limited, compliance-weighted estimates can in practice be interpreted as approximations to the LATE.

    \section{Conclusion}
\label{sec:conclusion}

In this article, we contribute to a recent literature in econometrics \citep{coussens2018, abadie2019instrumental,Huntington-Klein,Borusyak2020-xu} and build upon prior results in biostatistics \citep{Joffe2003-uq} that consider improving IV estimation through weighting (or selecting) samples according to the probability of compliance. Specifically, we relate the construction of interacted instruments to weighting in the IV estimator, consider a potential outcomes framework with a nonparametric first stage to clarify the resulting estimands, and derive the large-sample properties of the corresponding estimators. Moreover, we provide tools for valid inference even when instruments are estimated using nonparametric techniques such as machine learning algorithms for the estimation of heterogeneous treatment effects. We then document the improvements in precision that can be obtained from using these methods in a theoretical illustration and a simulation study.
We hope that the framework deployed in this article will provide practitioners with helpful insights and practical methods to improve inference in the often-encountered scenarios in which conventional IV estimates are noisy.
    
    \newpage
    
    \begin{graphstables}
        \clearpage
\newpage

\newcommand{\figfoot}{\floatfoot{Estimates were generated using 1,000 Monte Carlo samples with 1,000 observations each. `X-' prefix indicates that estimates were cross-fitted.}}

\begin{figure}
	\caption{Oracle compliance weights}
	\label{fig:w_dist}
	\centerline{\includegraphics[width=\textwidth]{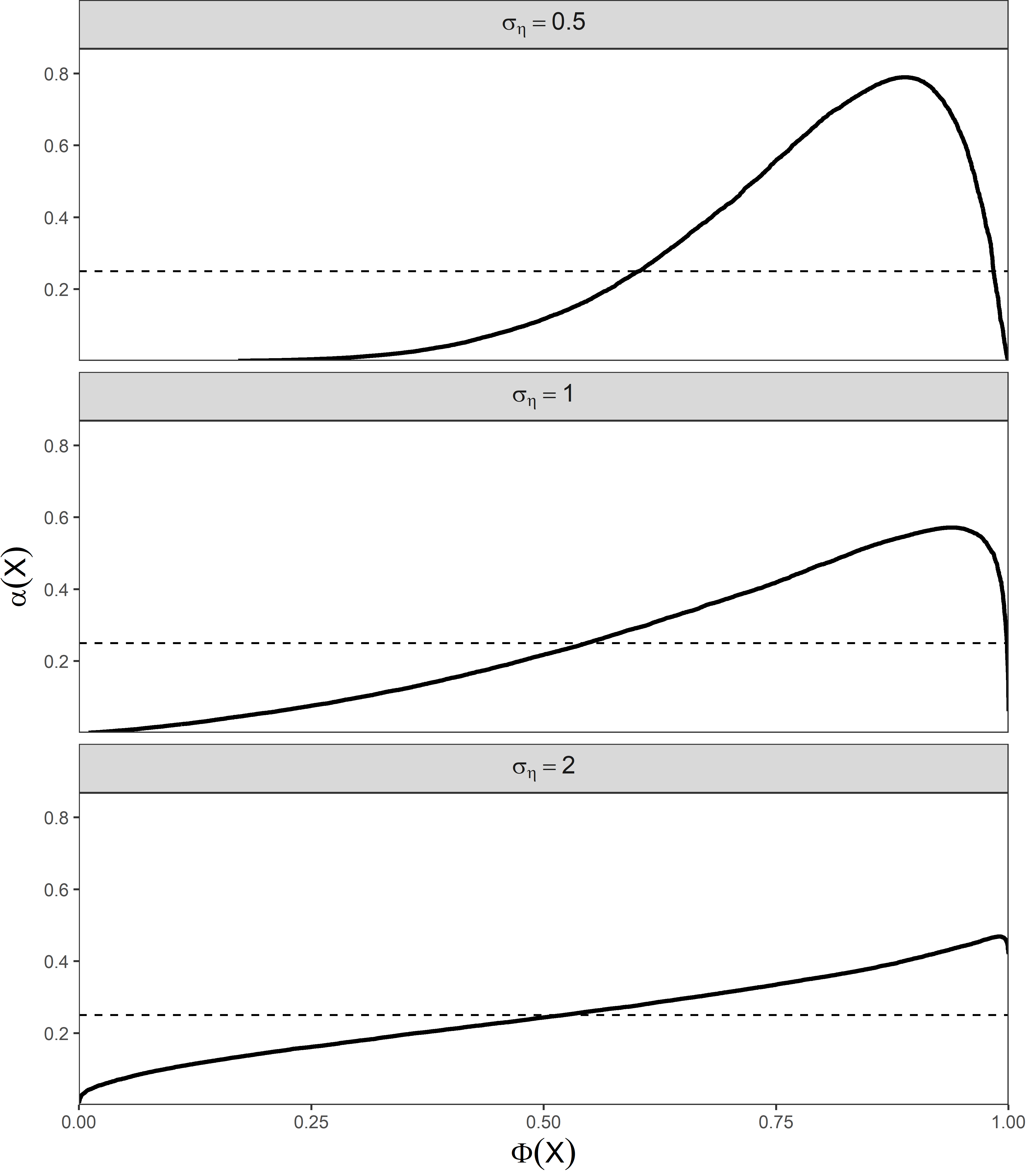}}
	\floatfoot{The oracle distribution of compliance weights is analytically derived (see \autoref{sec:sim_deriv} for details). The dashed horizontal line represents the population average compliance rate (0.25).}
\end{figure}

\begin{figure}
	\caption{Distribution of estimates relative to LATE}
	\label{fig:density_no_ols}
	\centerline{\includegraphics[width=\textwidth]{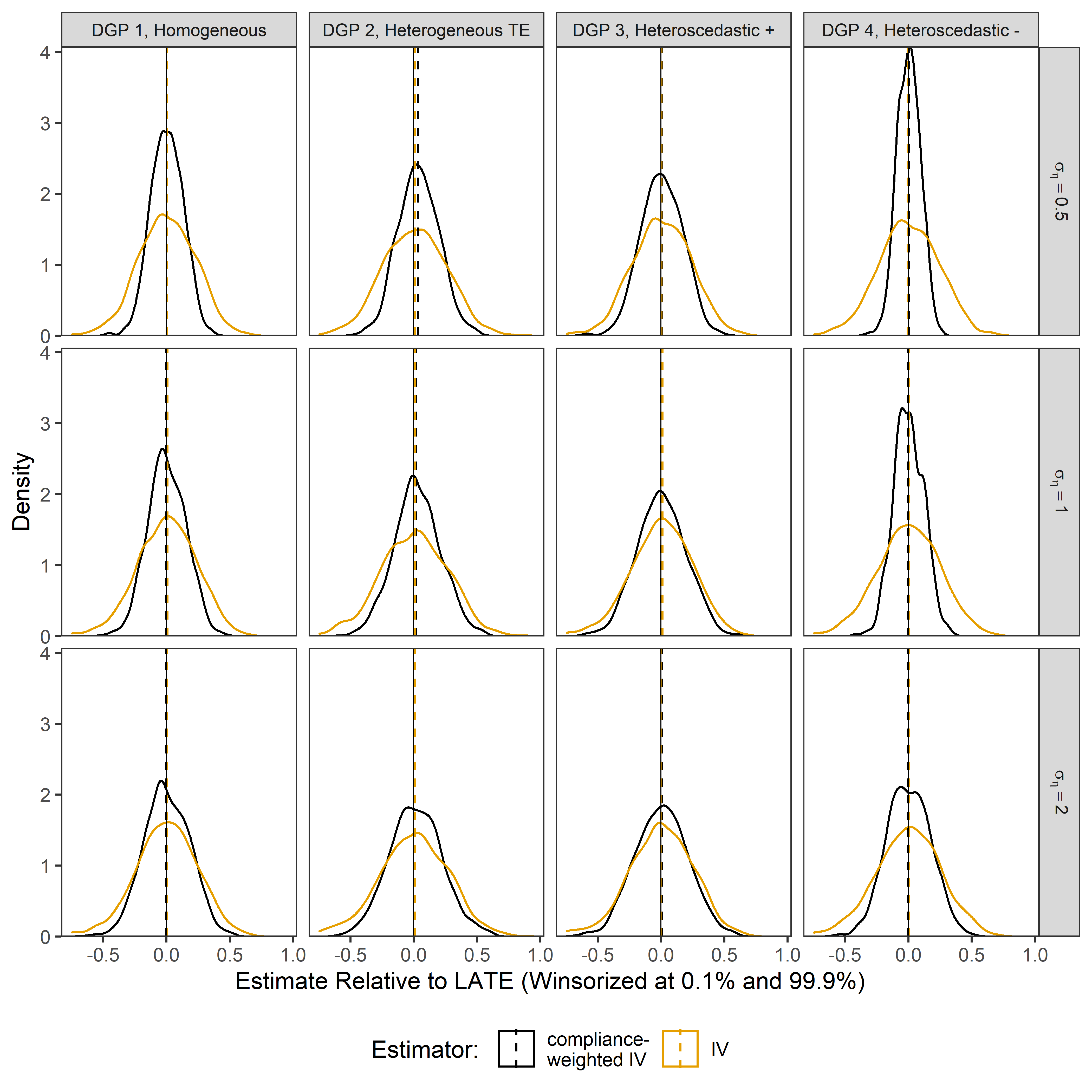}}
	\floatfoot{Distribution of estimates were generated using 1,000 Monte Carlo samples with 1,000 observations each, and oracle weights. The dashed vertical lines represent the median for each estimator.}
\end{figure}

\begin{figure}
	\caption{Compliance-weighted estimator precision relative to unweighted IV}
	\label{fig:rmse}
	\centerline{\includegraphics[width=\textwidth]{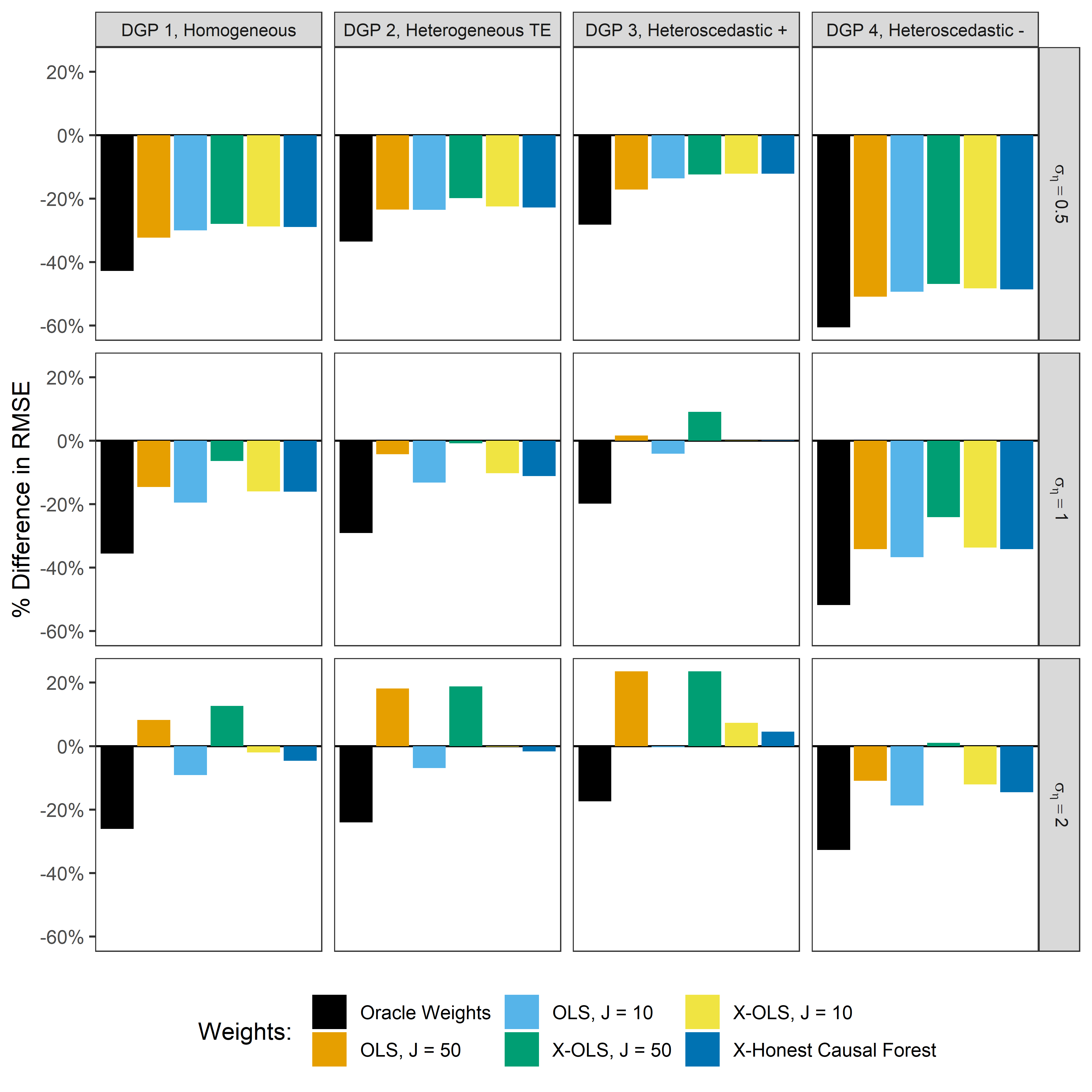}}
	\figfoot
\end{figure}

\begin{figure}
	\caption{MSE decomposition of compliance-weighted estimator relative to unweighted IV}
	\label{fig:decomp_mse}
	\centerline{\includegraphics[width=\textwidth]{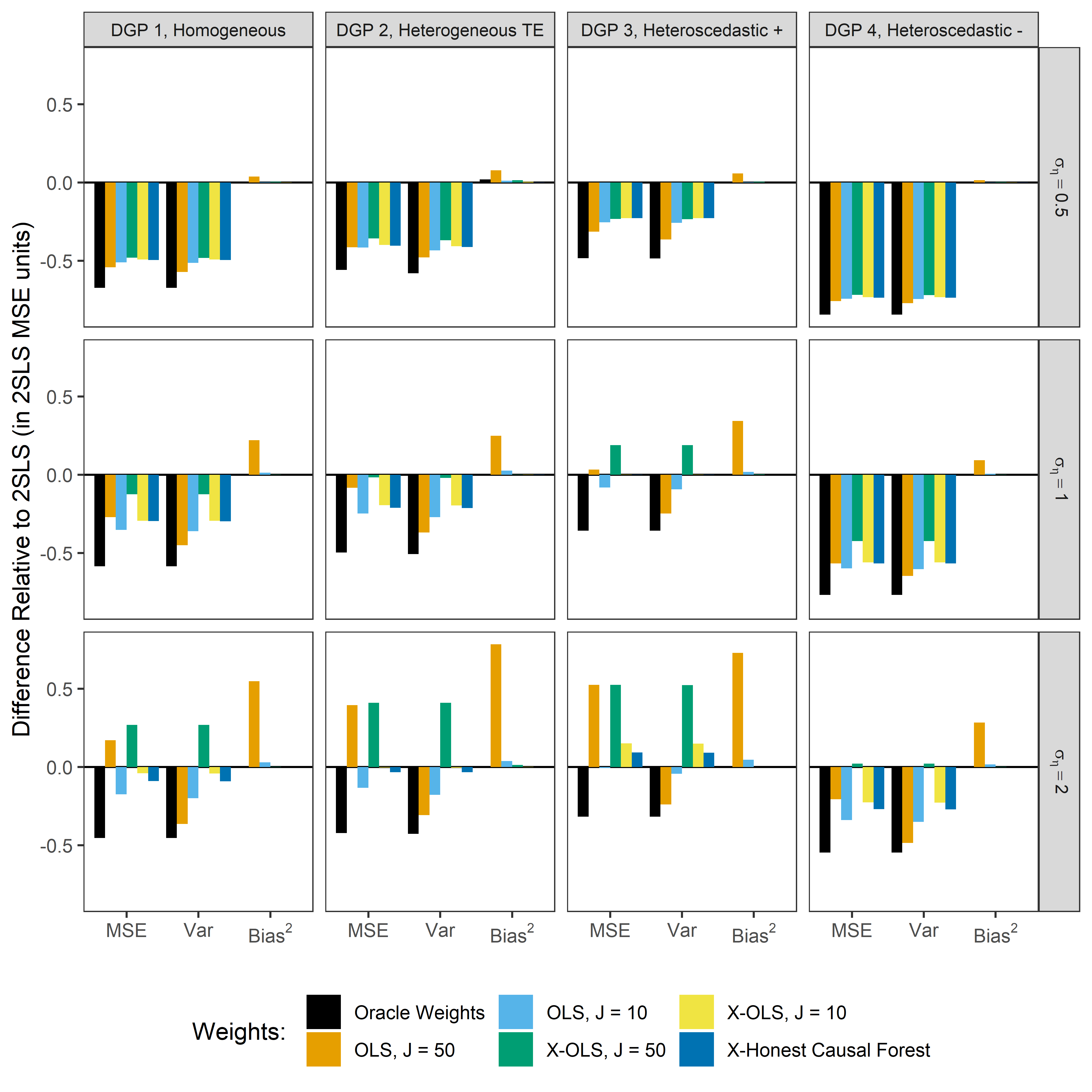}}
	\figfoot
\end{figure}

\clearpage
        \newpage
\clearpage

\newcommand{\tabfoot}{\captionsetup{width=0.8\textwidth}\floatfoot{Estimates were generated using 1,000 Monte Carlo samples with 1,000 observations.\newline
N/A: Unweighted; OW: Oracle Weights; J10: OLS, 10 bins; J50: OLS, 50 bins; HCF: Honest Causal Forest.
\textbf{Note:} `X-' prefix indicates that weight estimates are cross-fitted.}}

\setlength{\tabcolsep}{5pt} % Default value: 6pt

\begin{table}
	\caption{Estimator RMSE}
	\label{tab:rmse}
	\captionsetup[table]{labelformat=empty,skip=1pt}
\begin{longtable*}{cccccccccccc}
\toprule
$\sigma_{\eta}$ & $\sigma_{\tau}$ & $\rho_{\delta\tau}$ & $\rho_{\delta\varepsilon}$ & $\zeta$ & N/A & OW & J10 & J50 & X-J10 & X-J50 & X-HCF \\ 
\midrule
\multicolumn{10}{l}{DGP 1: Homogeneous} \\ 
\midrule
0.5 & 0 & 0.0 & 0.5 & 0.00 & $0.224$ & $0.128$ & $0.157$ & $0.152$ & $0.160$ & $0.161$ & $0.159$ \\ 
1.0 & 0 & 0.0 & 0.5 & 0.00 & $0.234$ & $0.151$ & $0.189$ & $0.200$ & $0.197$ & $0.219$ & $0.197$ \\ 
2.0 & 0 & 0.0 & 0.5 & 0.00 & $0.245$ & $0.181$ & $0.222$ & $0.265$ & $0.240$ & $0.276$ & $0.233$ \\ 
\midrule
\multicolumn{10}{l}{DGP 2: Heterogeneous TE} \\ 
\midrule
0.5 & 1 & 0.5 & 0.5 & 0.00 & $0.250$ & $0.166$ & $0.191$ & $0.191$ & $0.194$ & $0.201$ & $0.193$ \\ 
1.0 & 1 & 0.5 & 0.5 & 0.00 & $0.263$ & $0.187$ & $0.229$ & $0.252$ & $0.236$ & $0.261$ & $0.234$ \\ 
2.0 & 1 & 0.5 & 0.5 & 0.00 & $0.274$ & $0.208$ & $0.255$ & $0.324$ & $0.274$ & $0.326$ & $0.270$ \\ 
\midrule
\multicolumn{10}{l}{DGP 3: Heteroscedastic $+$} \\ 
\midrule
0.5 & 0 & 0.0 & 0.5 & 0.25 & $0.231$ & $0.166$ & $0.200$ & $0.191$ & $0.203$ & $0.203$ & $0.203$ \\ 
1.0 & 0 & 0.0 & 0.5 & 0.25 & $0.243$ & $0.195$ & $0.233$ & $0.247$ & $0.243$ & $0.265$ & $0.243$ \\ 
2.0 & 0 & 0.0 & 0.5 & 0.25 & $0.253$ & $0.209$ & $0.253$ & $0.313$ & $0.272$ & $0.313$ & $0.265$ \\ 
\midrule
\multicolumn{10}{l}{DGP 4: Heteroscedastic $-$} \\ 
\midrule
0.5 & 0 & 0.0 & 0.5 & -0.25 & $0.236$ & $0.093$ & $0.120$ & $0.116$ & $0.122$ & $0.126$ & $0.121$ \\ 
1.0 & 0 & 0.0 & 0.5 & -0.25 & $0.247$ & $0.119$ & $0.157$ & $0.163$ & $0.164$ & $0.188$ & $0.163$ \\ 
2.0 & 0 & 0.0 & 0.5 & -0.25 & $0.258$ & $0.174$ & $0.210$ & $0.230$ & $0.227$ & $0.261$ & $0.220$ \\ 
\bottomrule
\end{longtable*}

	\tabfoot % standard table notes about weight types
\end{table}

\begin{table}
	\caption{Estimator standard deviation}
	\label{tab:sd}
	\captionsetup[table]{labelformat=empty,skip=1pt}
\begin{longtable*}{cccccccccccc}
\toprule
$\sigma_{\eta}$ & $\sigma_{\tau}$ & $\rho_{\delta\tau}$ & $\rho_{\delta\varepsilon}$ & $\zeta$ & N/A & OW & J10 & J50 & X-J10 & X-J50 & X-HCF \\ 
\midrule
\multicolumn{10}{l}{DGP 1: Homogeneous} \\ 
\midrule
0.5 & 0 & 0.0 & 0.5 & 0.00 & $0.224$ & $0.128$ & $0.156$ & $0.147$ & $0.160$ & $0.161$ & $0.159$ \\ 
1.0 & 0 & 0.0 & 0.5 & 0.00 & $0.234$ & $0.151$ & $0.188$ & $0.174$ & $0.197$ & $0.219$ & $0.197$ \\ 
2.0 & 0 & 0.0 & 0.5 & 0.00 & $0.245$ & $0.181$ & $0.219$ & $0.195$ & $0.240$ & $0.276$ & $0.233$ \\ 
\midrule
\multicolumn{10}{l}{DGP 2: Heterogeneous TE} \\ 
\midrule
0.5 & 1 & 0.5 & 0.5 & 0.00 & $0.250$ & $0.162$ & $0.188$ & $0.180$ & $0.192$ & $0.199$ & $0.192$ \\ 
1.0 & 1 & 0.5 & 0.5 & 0.00 & $0.263$ & $0.185$ & $0.225$ & $0.209$ & $0.236$ & $0.261$ & $0.234$ \\ 
2.0 & 1 & 0.5 & 0.5 & 0.00 & $0.274$ & $0.208$ & $0.249$ & $0.228$ & $0.274$ & $0.326$ & $0.270$ \\ 
\midrule
\multicolumn{10}{l}{DGP 3: Heteroscedastic +} \\ 
\midrule
0.5 & 0 & 0.0 & 0.5 & 0.25 & $0.231$ & $0.166$ & $0.199$ & $0.184$ & $0.203$ & $0.202$ & $0.203$ \\ 
1.0 & 0 & 0.0 & 0.5 & 0.25 & $0.243$ & $0.195$ & $0.231$ & $0.211$ & $0.243$ & $0.265$ & $0.243$ \\ 
2.0 & 0 & 0.0 & 0.5 & 0.25 & $0.254$ & $0.209$ & $0.248$ & $0.221$ & $0.272$ & $0.313$ & $0.265$ \\ 
\midrule
\multicolumn{10}{l}{DGP 4: Heteroscedastic -} \\ 
\midrule
0.5 & 0 & 0.0 & 0.5 & -0.25 & $0.236$ & $0.093$ & $0.120$ & $0.113$ & $0.122$ & $0.125$ & $0.121$ \\ 
1.0 & 0 & 0.0 & 0.5 & -0.25 & $0.247$ & $0.119$ & $0.156$ & $0.147$ & $0.164$ & $0.188$ & $0.163$ \\ 
2.0 & 0 & 0.0 & 0.5 & -0.25 & $0.258$ & $0.174$ & $0.208$ & $0.185$ & $0.227$ & $0.261$ & $0.220$ \\ 
\bottomrule
\end{longtable*}

	\tabfoot % standard table notes about weight types
\end{table}

\begin{table}
	\caption{Estimator bias}
	\label{tab:bias}
	\captionsetup[table]{labelformat=empty,skip=1pt}
\begin{longtable*}{cccccccccccc}
\toprule
$\sigma_{\eta}$ & $\sigma_{\tau}$ & $\rho_{\delta\tau}$ & $\rho_{\delta\varepsilon}$ & $\zeta$ & N/A & OW & J10 & J50 & X-J10 & X-J50 & X-HCF \\ 
\midrule
\multicolumn{10}{l}{DGP 1: Homogeneous} \\ 
\midrule
0.5 & 0 & 0.0 & 0.5 & 0.00 & $0.002$ & $0.005$ & $0.015$ & $0.044$ & $0.008$ & $0.017$ & $0.005$ \\ 
1.0 & 0 & 0.0 & 0.5 & 0.00 & $0.008$ & $-0.004$ & $0.027$ & $0.110$ & $0.004$ & $0.007$ & $0.001$ \\ 
2.0 & 0 & 0.0 & 0.5 & 0.00 & $0.006$ & $-0.005$ & $0.042$ & $0.181$ & $-0.002$ & $0.007$ & $0.000$ \\ 
\midrule
\multicolumn{10}{l}{DGP 2: Heterogeneous TE} \\ 
\midrule
0.5 & 1 & 0.5 & 0.5 & 0.00 & $0.008$ & $0.035$ & $0.028$ & $0.070$ & $0.017$ & $0.030$ & $0.013$ \\ 
1.0 & 1 & 0.5 & 0.5 & 0.00 & $0.012$ & $0.019$ & $0.043$ & $0.132$ & $0.008$ & $0.011$ & $0.013$ \\ 
2.0 & 1 & 0.5 & 0.5 & 0.00 & $0.010$ & $0.012$ & $0.053$ & $0.243$ & $-0.008$ & $0.031$ & $0.000$ \\ 
\midrule
\multicolumn{10}{l}{DGP 3: Heteroscedastic $+$} \\ 
\midrule
0.5 & 0 & 0.0 & 0.5 & 0.25 & $0.003$ & $0.004$ & $0.014$ & $0.056$ & $0.002$ & $0.017$ & $0.004$ \\ 
1.0 & 0 & 0.0 & 0.5 & 0.25 & $0.011$ & $0.000$ & $0.033$ & $0.142$ & $0.002$ & $0.012$ & $0.001$ \\ 
2.0 & 0 & 0.0 & 0.5 & 0.25 & $0.000$ & $0.008$ & $0.054$ & $0.217$ & $-0.003$ & $-0.002$ & $-0.006$ \\ 
\midrule
\multicolumn{10}{l}{DGP 4: Heteroscedastic $-$} \\ 
\midrule
0.5 & 0 & 0.0 & 0.5 & -0.25 & $-0.005$ & $0.003$ & $0.011$ & $0.029$ & $0.008$ & $0.012$ & $0.009$ \\ 
1.0 & 0 & 0.0 & 0.5 & -0.25 & $0.005$ & $0.000$ & $0.020$ & $0.075$ & $0.007$ & $0.008$ & $0.005$ \\ 
2.0 & 0 & 0.0 & 0.5 & -0.25 & $0.007$ & $-0.002$ & $0.033$ & $0.138$ & $-0.004$ & $0.009$ & $-0.003$ \\ 
\bottomrule
\end{longtable*}

	\tabfoot % standard table notes about weight types
\end{table}

\begin{table}
	\caption{LATE coverage of compliance-weighted estimator}
	\label{tab:coverage_late}
	\captionsetup[table]{labelformat=empty,skip=1pt}
\begin{longtable*}{ccccccccccc}
\toprule
$\sigma_{\eta}$ & $\sigma_{\tau}$ & $\rho_{\delta\tau}$ & $\rho_{\delta\varepsilon}$ & $\zeta$ & OW & J10 & J50 & X-J10 & X-J50 & X-HCF \\ 
\midrule
\multicolumn{10}{l}{DGP 1: Homogeneous} \\ 
\midrule
0.5 & 0 & 0.0 & 0.5 & 0.00 & $0.964$ & $0.948$ & $0.947$ & $0.950$ & $0.952$ & $0.953$ \\ 
1.0 & 0 & 0.0 & 0.5 & 0.00 & $0.952$ & $0.942$ & $0.919$ & $0.952$ & $0.952$ & $0.950$ \\ 
2.0 & 0 & 0.0 & 0.5 & 0.00 & $0.961$ & $0.941$ & $0.846$ & $0.950$ & $0.951$ & $0.951$ \\ 
\midrule
\multicolumn{10}{l}{DGP 2: Heterogeneous TE} \\ 
\midrule
0.5 & 1 & 0.5 & 0.5 & 0.00 & $0.942$ & $0.939$ & $0.935$ & $0.942$ & $0.949$ & $0.944$ \\ 
1.0 & 1 & 0.5 & 0.5 & 0.00 & $0.946$ & $0.942$ & $0.897$ & $0.948$ & $0.947$ & $0.946$ \\ 
2.0 & 1 & 0.5 & 0.5 & 0.00 & $0.951$ & $0.954$ & $0.812$ & $0.955$ & $0.955$ & $0.954$ \\ 
\midrule
\multicolumn{10}{l}{DGP 3: Heteroscedastic $+$} \\ 
\midrule
0.5 & 0 & 0.0 & 0.5 & 0.25 & $0.958$ & $0.948$ & $0.941$ & $0.950$ & $0.954$ & $0.952$ \\ 
1.0 & 0 & 0.0 & 0.5 & 0.25 & $0.954$ & $0.947$ & $0.908$ & $0.950$ & $0.954$ & $0.941$ \\ 
2.0 & 0 & 0.0 & 0.5 & 0.25 & $0.959$ & $0.936$ & $0.809$ & $0.949$ & $0.955$ & $0.956$ \\ 
\midrule
\multicolumn{10}{l}{DGP 4: Heteroscedastic $-$} \\ 
\midrule
0.5 & 0 & 0.0 & 0.5 & -0.25 & $0.954$ & $0.953$ & $0.950$ & $0.951$ & $0.950$ & $0.957$ \\ 
1.0 & 0 & 0.0 & 0.5 & -0.25 & $0.967$ & $0.952$ & $0.935$ & $0.955$ & $0.961$ & $0.952$ \\ 
2.0 & 0 & 0.0 & 0.5 & -0.25 & $0.960$ & $0.949$ & $0.895$ & $0.952$ & $0.956$ & $0.948$ \\ 
\bottomrule
\end{longtable*}

	\tabfoot % standard table notes about weight types
\end{table}

\begin{table}
	\caption{Weighted LATE coverage of compliance-weighted estimator}
	\label{tab:coverage_wlate}
	\captionsetup[table]{labelformat=empty,skip=1pt}
\begin{longtable*}{ccccccccccc}
\toprule
$\sigma_{\eta}$ & $\sigma_{\tau}$ & $\rho_{\delta\tau}$ & $\rho_{\delta\varepsilon}$ & $\zeta$ & OW & J10 & J50 & X-J10 & X-J50 & X-HCF \\ 
\midrule
\multicolumn{10}{l}{DGP 1: Homogeneous} \\ 
\midrule
0.5 & 0 & 0.0 & 0.5 & 0.00 & $0.964$ & $0.948$ & $0.947$ & $0.950$ & $0.952$ & $0.953$ \\ 
1.0 & 0 & 0.0 & 0.5 & 0.00 & $0.952$ & $0.942$ & $0.919$ & $0.952$ & $0.952$ & $0.950$ \\ 
2.0 & 0 & 0.0 & 0.5 & 0.00 & $0.961$ & $0.941$ & $0.846$ & $0.950$ & $0.951$ & $0.951$ \\ 
\midrule
\multicolumn{10}{l}{DGP 2: Heterogeneous TE} \\ 
\midrule
0.5 & 1 & 0.5 & 0.5 & 0.00 & $0.952$ & $0.946$ & $0.947$ & $0.943$ & $0.950$ & $0.948$ \\ 
1.0 & 1 & 0.5 & 0.5 & 0.00 & $0.955$ & $0.950$ & $0.909$ & $0.948$ & $0.948$ & $0.953$ \\ 
2.0 & 1 & 0.5 & 0.5 & 0.00 & $0.950$ & $0.953$ & $0.822$ & $0.957$ & $0.955$ & $0.955$ \\ 
\midrule
\multicolumn{10}{l}{DGP 3: Heteroscedastic $+$} \\ 
\midrule
0.5 & 0 & 0.0 & 0.5 & 0.25 & $0.958$ & $0.948$ & $0.941$ & $0.950$ & $0.954$ & $0.952$ \\ 
1.0 & 0 & 0.0 & 0.5 & 0.25 & $0.954$ & $0.947$ & $0.908$ & $0.950$ & $0.954$ & $0.941$ \\ 
2.0 & 0 & 0.0 & 0.5 & 0.25 & $0.959$ & $0.936$ & $0.809$ & $0.949$ & $0.955$ & $0.956$ \\ 
\midrule
\multicolumn{10}{l}{DGP 4: Heteroscedastic $-$} \\ 
\midrule
0.5 & 0 & 0.0 & 0.5 & -0.25 & $0.954$ & $0.953$ & $0.950$ & $0.951$ & $0.950$ & $0.957$ \\ 
1.0 & 0 & 0.0 & 0.5 & -0.25 & $0.967$ & $0.952$ & $0.935$ & $0.955$ & $0.961$ & $0.952$ \\ 
2.0 & 0 & 0.0 & 0.5 & -0.25 & $0.960$ & $0.949$ & $0.895$ & $0.952$ & $0.956$ & $0.948$ \\ 
\bottomrule
\end{longtable*}

	\tabfoot % standard table notes about weight types
\end{table}

\begin{table}
	\caption{Mean first stage estimates}
	\label{tab:fs}
	\captionsetup[table]{labelformat=empty,skip=1pt}
\begin{longtable*}{cccccccc}
\toprule
$\sigma_{\eta}$ & N/A & OW & J10 & J50 & X-J10 & X-J50 & X-HCF \\ 
\midrule
0.5 & $0.248$ & $0.643$ & $0.539$ & $0.578$ & $0.529$ & $0.534$ & $0.529$ \\ 
1.0 & $0.248$ & $0.431$ & $0.394$ & $0.439$ & $0.378$ & $0.370$ & $0.380$ \\ 
2.0 & $0.248$ & $0.304$ & $0.310$ & $0.368$ & $0.292$ & $0.286$ & $0.293$ \\ 
\bottomrule
\end{longtable*}

	\tabfoot % standard table notes about weight types
\end{table}

\newpage
\clearpage

    \end{graphstables}

    % \nocite{*}
    \bibliography{ref}

    \begin{vappendix}
        \newpage
        \appendix
        \counterwithin{proposition}{section}
        \counterwithin{corollary}{section}
        \counterwithin{theorem}{section}
        \counterwithin{conjecture}{section}

        \begin{center}
            \LARGE Appendix
        \end{center}

        \section{Generalization to Non-Binary Treatment and Instrument}

\label{sec:generallinear}

In \autoref{sec:generalbinary}, we considered the general large-sample limit of the $w$-weighted IV estimator $\hat{\tau}_w$ in the case of binary instrument and binary treatment, which we obtained from the instrument $h(Z,X) = w(X) (Z-p)$.
In this section, we expand the theory in two directions.
First, we argue that the main asymptotic results extend to the case of a continuous, rather than only binary, treatment variable in a linear model.
Second, we consider instruments that are not binary, and ask which transformations $h(Z,X)$ generally provide an instrument that remains valid.

\autoref{sec:generalbinary} focused on estimating a weighted average of the conditional LATEs
\begin{align*}
    \tau(x) = \E[Y(1) - Y(0) | D(1) > D(0), X=x].
\end{align*}
However, when treatment $D$ or instrument $Z$ are not binary,
we lose the estimand's intuitive interpretation as the average treatment effect among compliers. Nonetheless, the results above still apply in a natural generalization:

\begin{proposition}[Generalization for continuous treatment]
    \label{prop:generalization}
    The asymptotic results of the previous section generalize to continuous treatment $D$ (and for continuous instruments $Z$) under standard regularity assumptions,
    where we define
    \begin{align*}
        \tau(x) &= \frac{\Cov(Y,Z|X=x)}{\Cov(D,Z|X=x)},
        &
        \alpha(x) &= \frac{\Cov(D,Z|X=x)}{\Var(Z)}
    \end{align*}
    and replace the conditional distribution $\P(\cdot|D(1)>D(0))$ of compliers by the $\alpha$-weighted distribution.
    Specifically,
    \begin{align*}
        \tau &= \frac{\Cov(Y,Z)}{\Cov(D,Z)} = \frac{\E[\alpha(X)\: \tau(X)]}{\E[\alpha(X)]},
        &
        \tau_w &= \frac{\E[\alpha(X)\: w(X) \: \tau(X)]}{\E[\alpha(X) \: w(X)]},
        &
        \tau_\alpha &= \frac{\E[(\alpha(X))^2 \: \tau(X)]}{\E[(\alpha(X))^2]}.
    \end{align*}
\end{proposition}

In \autoref{prop:binarycharacterization} we characterized a class of augmented instruments $h(Z,X)$ that remains valid no matter the distribution of covariates $X$ and error term $\varepsilon$ in the second-stage equation $Y = \psi + D \tau + \varepsilon$ for which $Z$ is independent of $(X,\varepsilon)$, for binary $Z$. While the resulting instruments $h(Z,X) = w(X) \: (Z-p)$ remain valid when $Z$ is not binary, there are additional transformations that do the trick for more general instruments:

\begin{proposition}[Characterization of augmented extensions of a finitely supported instrument]
    \label{prop:generalcharacterization}
    Assume that $Z$ has finite support $\Z$ with $|\Z|=J+1$ and fix the distribution of $Z$.
    Consider some conformal function $h(\cdot,\cdot)$.
    Then $\Cov(\varepsilon,h(Z,X)) = 0$ for all distributions of $(X,\varepsilon)$ with $\Var(h(Z,X)), \Var(\varepsilon) < \infty$ and $(X,\varepsilon) \orth Z$ if and only if
    \begin{align*}
        h(z,x)
        =
        \psi +
        \sum_{j=1}^J g_j(z) \: w_j(x)
    \end{align*}
    with $\E[g_j(Z)] = 0, \E[g_{j}(Z) g_{j'}(Z)] = \Ind_{j = j'}$ and some constant $\psi$, where we can choose a family of $g_j$ that only depends on the distribution of $Z$.
\end{proposition}

We think of the restriction to finite support as an approximation to arbitrary distributions; while there may be additional instruments when the space of random variables is more complex, simple averages of the above form remain valid.
The result is a direct generalization of \autoref{prop:binarycharacterization}, which we can obtain by setting $g(Z) = \frac{Z-p}{\sqrt{p(1-p)}}$.

Defining
\begin{align*}
    \tau_j(x) &= \frac{\Cov(Y,g_{j}(Z)|X=x)}{\Cov(D,g_{j}(Z)|X=x)},
    &
    \alpha_j(x) &= \Cov(D,g_{j}(Z)|X=x)
\end{align*}
for the treatment effect and compliance score associated with the basis instrument $g_j(Z)$,
the instrument in \autoref{prop:generalcharacterization}
under regularity conditions is consistent for the weighted treatment effect
\begin{align*}
    \tau_w = \frac{\sum_{j=1}^J \E\left[ \alpha_j(X) \: w_j(X) \: \tau_j(X) \right]}{\sum_{j=1}^J \E\left[ \alpha_j(X) \: w_j(X) \right]}.
\end{align*}
In the homoscedastic case, the IV estimator with instrument $h(Z,X)$ is also equivalent to the convex combination of weighted IV estimators
\begin{align*}
    \hat{\tau}_w =
    \frac{
        \sum_{j=1}^J \sCov^{w_j}(D,g_j(Z)) \: \frac{\sCov^{w_j}(Y,g_j(Z))}{\sCov^{w_j}(D,g_j(Z))}
    }{
        \sum_{j=1}^J \sCov^{w_j}(D,g_j(Z))
    },
\end{align*}
when we include the weights $w_j(X)$ as controls.

The efficient choice of weights are compliance weights $w_j(x) =\alpha_j(x)$ when errors are homoscedastic and treatment effects homogeneous.
We obtain an estimator of a compliance-weighted average of compliance-weighted treatment effects, 
\begin{align*}
    \tau^* &= \sum_{j=1}^J \overline{\alpha}_j \: \tau^*_j,
    &
    \overline{\alpha}_j &= \frac{\E\left[ (\alpha_j(X))^2\right]}{\sum_{j'=1}^J \E\left[ (\alpha_{j'}(X))^2\right]},
    &
    \tau^*_j
    &=
    \frac{\E\left[ (\alpha_j(X))^2 \: \tau_j(X) \right]}{\E\left[ (\alpha_j(X))^2\right]}.
\end{align*}

While we call $\tau^*,\tau^*_j, \tau_j(x)$ treatment effects here, we note that the causal interpretation of such quantities is not straightforward.
Indeed, \citet{Mogstad2019-wa,Mogstad_undated-it} argue that estimating causal effects with multiple instruments requires more careful statements of assumptions and implied estimands, as well as alternative estimation strategies that go beyond simple (weighted) averages of two-stage least-squares estimates.

\section{A General Result for Weighted Estimation}

In this section, we provide a general result on the large-sample distribution of weighted sums, where weights are fitted by $k$-fold cross-fitting, but may converge slower than at the parametric rate.

\begin{lemma}
    \label{lem:weightedclt}
    Consider weights $W_i = \hat{w}_{-j(i)}(X_i)$, where $\hat{w}_{-j}$ are fitted by $k$-fold cross-fitting excluding the $j$-th fold and $j(i)$ denotes the fold of observation $i$.
    Consider some random variable $A$, and let $a(x) = \E[A|X=x]$.
    We assume that there is some function $w(x)$ with finite variance such that
    $\E[(\hat{w}_{-j}(X) - w(X))^2] \rightarrow 0$, where the expectation is over the training sample for $\hat{w}_{-j}$ as well as an independent draw $X$.
    Also assume that $a(X), \Var(A|X) \leq C$ almost surely for some constant $C$.
    Then
    \begin{align*}
        \frac{1}{n} \sum_{i=1}^n W_i A_i
        \cp
        \E[w(X) \: a(X)]
    \end{align*}
    and
    \begin{align*}
        \frac{1}{\sqrt{n}} \left(\sum_{i=1}^n W_i A_i - \sum_{i=1}^n W_i a(X_i)\right)
        &\cd
        \N(0,\E[w^2(X) \Var(A|X)]),
        \\
        \sqrt{n} \left(\frac{1}{n} \sum_{i=1}^n W_i A_i - \frac{1}{k} \sum_{j=1}^k \E[\hat{w}_{-j}(X) \: a(X)|\hat{w}_{-j}]\right)
        &\cd
        \N(0,\underbrace{\Var(w(X) A)}_{\mathclap{
            = \E[w^2(X) \Var(A|X)] + \Var(w(X) a(X))
        }}).
    \end{align*}
\end{lemma}

\begin{proof}[Proof of \autoref{lem:weightedclt}]
Write $\varepsilon = A - a(X)$.
For the first asymptotic Normality result, we find that
\begin{align*}
    \frac{1}{\sqrt{n}} \left(\sum_{i=1}^n W_i A_i - \sum_{i=1}^n W_i a(X_i)\right)
    &=
    \frac{1}{\sqrt{n}} \sum_{i=1}^n W_i \varepsilon_i
    \\
    &= \frac{1}{\sqrt{n}} \sum_{i=1}^n w(X_i) \varepsilon_i
    +  \frac{1}{\sqrt{n}} \sum_{i=1}^n (\hat{w}_{-j(i)}(X_i) - w(X_i)) \varepsilon_i
\end{align*}
where
\begin{align*}
    \frac{1}{\sqrt{n}} \sum_{i=1}^n w(X_i) \varepsilon_i
    \cd \N(0,\E[w^2(X) \Var(A|X)])
\end{align*}
and (ignoring rounding issues for fold sizes)
\begin{align}
    \label{eqn:residualsmall}
    \begin{split}
    &\E\left[
        \left(
            \frac{1}{\sqrt{n}} \sum_{i=1}^n (\hat{w}_{-j(i)}(X_i) - w(X_i)) \varepsilon_i
        \right)^2
    \right]
    =
    \Var\left(\frac{1}{\sqrt{n}} \sum_{i=1}^n (\hat{w}_{-j(i)}(X_i) - w(X_i)) \varepsilon_i
    \right)
    \\
    &\leq
    \frac{k^2}{n} \Var\left( \sum_{i; j(i) = 1} (\hat{w}_{-j}(X_i) - w(X_i)) \varepsilon_i
    \right)
    =
    \frac{k^2}{n} \E\left[\Var\left( \sum_{i; j(i) = 1} (\hat{w}_{-j}(X_i) - w(X_i)) \varepsilon_i
    \middle|\hat{w}_{-j}\right)\right]
    \\
    &=
    \frac{k^2}{n} \frac{n}{k} \E\left[(\hat{w}_{-j}(X) - w(X))^2 \Var(A|X)\right]
    = k \E\left[(\hat{w}_{-j}(X) - w(X))^2 \Var(A|X)\right]
    \\
    &\leq k \: C \: \E\left[(\hat{w}_{-j}(X) - w(X))^2\right]
   \rightarrow 0.
    \end{split}
\end{align}
For the second asymptotic Normality result,
\begin{align*}
    &\sqrt{n} \left(\frac{1}{n} \sum_{i=1}^n W_i A_i - \frac{1}{k} \sum_{j=1}^k \E[\hat{w}_{-j}(X) \: a(X)|\hat{w}_{-j}]\right)
    \\
    &=
    \sqrt{n} \left(\frac{1}{n} \sum_{j=1}^k \sum_{i;j(i)=j} (\hat{w}_{-j}(X_i) A_i - \E[\hat{w}_{-j}(X) \: a(X)|\hat{w}_{-j}])\right)
    \\
    &=
    \frac{1}{\sqrt{n}} \sum_{i=1}^n (w(X) A - \E[w(X) a(X)])
    \\
    &\phantom{=}+ \sum_{j=1}^k \frac{1}{\sqrt{n}}  \sum_{i;j(i)=j} ((\hat{w}_{-j}(X_i) {-} w(X_i)) A_i - \E[(\hat{w}_{-j}(X) {-} w(X)) \: a(X)|\hat{w}_{-j}])
\end{align*}
with 
\begin{align*}
    \frac{1}{\sqrt{n}} \sum_{i=1}^n (w(X) A - \E[w(X) a(X)])
    \cd \N(0, \underbrace{\Var(w(X) A)}_{\mathclap{=\Var(w(X) a(X)) + \E[w^2(X) \Var(A|X)]}})
\end{align*}
and, for all $j$ almost surely,
\begin{align*}
    &\E\left[\frac{1}{\sqrt{n}}
    \sum_{i;j(i)=j}
    ((\hat{w}_{-j}(X_i) {-} w(X_i)) A_i - \E[(\hat{w}_{-j}(X) {-} w(X)) \: a(X)|\hat{w}_{-j}])\middle|\hat{w}_{-j}\right]
    = 0,
    \\
    &\Var\left(\frac{1}{\sqrt{n}}
    \sum_{i;j(i)=j}
    ((\hat{w}_{-j}(X_i) {-} w(X_i)) A_i - \E[(\hat{w}_{-j}(X) {-} w(X)) \: a(X)\middle|\hat{w}_{-j}])|\hat{w}_{-j}\right)
    \\
    &=
    \Var((\hat{w}_{-j}(X) {-} w(X)) A|\hat{w}_{-j}) / k
    \\
    &= \Var((\hat{w}_{-j}(X) {-} w(X)) a(X)|\hat{w}_{-j}) / k + \E[(\hat{w}_{-j}(X) {-} w(X))^2 \Var(A|X)|\hat{w}_{-j}] / k
    \\
    &\leq
    2 C \E[(\hat{w}_{-j}(X) {-} w(X))^2|\hat{w}_{-j}] / k.
\end{align*}
Hence,
\begin{align*}
    \E\left[\left(
        \sum_{j=1}^k \frac{1}{\sqrt{n}}  \sum_{i;j(i)=j} ((\hat{w}_{-j}(X_i) {-} w(X_i)) A_i - \E[(\hat{w}_{-j}(X) {-} w(X)) \: a(X)|\hat{w}_{-j}]
        \right)^2
    \right]
    \rightarrow 0.
\end{align*}

For consistency,
\begin{align*}
    &\E\left[\left(\frac{1}{n} \sum_{i=1}^n (W_i A_i - w(X_i) A_i)\right)^2\right]
    \leq 
    \frac{1}{n} \sum_{i=1}^n  \E[(W_i - w(X_i))^2 A^2_i]
    \\
    &=
    \E[(\hat{w}_{-j}(X) - w(X))^2 A^2]
    \leq 2 C \E[(\hat{w}_{-j}(X) - w(X))^2]
    \rightarrow 0,
\end{align*}
so $\frac{1}{n} \sum_{i=1}^n W_i A_i = \frac{1}{n}\sum_{i=1}^n w(X_i) A_i +  \frac{1}{n} \sum_{i=1}^n (W_i A_i - w(X_i) A_i) \stackrel{\mathcal{L}_2}{\rightarrow} \E[w(X) A] = \E[w(X) a(X)]$.
\end{proof}

\section{Proofs}

\begin{proof}[Proof of \autoref{prop:binarycharacterization}]
    Follows from the general result in \autoref{prop:generalcharacterization} with $g(Z) = \frac{Z-p}{\sqrt{p(1-p)}}$.
\end{proof}

\begin{proof}[Proof of \autoref{prop:optimalweight}]
    Writing $\alpha(X) = \E[D|Z=1,X] - \E[D|Z=0,X]$, we have that
    \begin{align*}
        &\E[(D - w(X) \: (Z-p))^2]
        \\
        &= \E^2[D] + \Var(D - w(X) \: (Z-p))
        \\
        &= \E^2[D] + \E[\Var(D|Z,X)] + \Var(\E[D|X,Z] - w(X) \: (Z-p))]
        \\
        &= 
        \text{const.}
        + \Var(\E[D|X])
        + \E[\Var(\E[D|X,Z] - w(X) \: (Z-p)|X)]
        \\
        &=
        \text{const.}
        + p (1-p) \E[((\E[D|X,Z=1] - w(X) \: (1-p)) - (\E[D|X,Z=0] + w(X) \: p))^2]
        \\
        &
        = 
        \text{const.} + p(1-p) \E\left[(\alpha(X) - w(X))^2 \right],
    \end{align*}
    which is minimal for $w(X) = \alpha(X)$.
\end{proof}

\begin{proof}[Proof of \autoref{prop:controlequivalence}]
    This follows from equivalent in the case of estimated weights in \autoref{prop:estimatedconsistent}.
\end{proof}

\begin{proof}[Proofs of \autoref{prop:weightedlate} and \autoref{prop:asymptotic}]
    Follow from the general result in \autoref{prop:estimatedconsistent} and \autoref{prop:estimatedasymptotic} that allow weights to be estimated.
\end{proof}

\begin{proof}[Proof of \autoref{prop:asymptoticlocal}]
    Along this asymptotic sequence,
    by \autoref{prop:asymptotic}
    \begin{align*}
        \sqrt{n} (\hat{\tau}_w - \tau)
        =
        \sqrt{n} (\hat{\tau}_w - \tau_w)
        +
        \underbrace{\sqrt{n} (\tau_w - \tau)}_{\mathclap{
            \substack{
            = 
            \frac{\E[\alpha(X) \: w(X) \: \sqrt{n}\tau(X)]}{\E[\alpha(X) \: w(X)]}
            -
            \frac{\E[\alpha(X) \: \sqrt{n}\tau(X)]}{\E[\alpha(X)]}
            \\
            =
            \frac{\E[\alpha(X) \: w(X) \: \mu(X)] \E[\alpha(X)] - \E[\alpha(X) \: \mu(X)] \E[\alpha(X) \: w(X)]}{\E[\alpha(X)] \E[\alpha(X) \: w(X)]}
            \\
            =
            \frac{\frac{\E[\alpha(X) \: w(X) \: \mu(X)]}{\E[\alpha(X)]} - \frac{\E[\alpha(X) \: \mu(X)]}{\E[\alpha(X)]} \frac{\E[\alpha(X) \: w(X)]}{\E[\alpha(X)]}}{\frac{\E[\alpha(X) \: w(X)]}{\E[\alpha(X)]}}
            }
        }}
        = \sqrt{n} (\hat{\tau}_w - \tau_w) + B_w
        \cd \N(B_w,V_w).
    \end{align*}
    Futhermore,
    considering residuals based on
    $Y - \tau_w D$ or $Y - \tau(X) D$ for the variance is both asymptotically equivalent since $\tau(X)$ and $\tau_w$ are local to zero, so the resulting variance $V_w$ does not depend on treatment effects.
\end{proof}

\begin{proof}[Proof of \autoref{cor:dominance}]
    Plugging into \autoref{eqn:asymptloss} (which is immediate from \autoref{prop:asymptoticlocal}),
    the estimator $\tau(\lambda)$
    has asymptotic loss
    \begin{align*}
        L(\lambda)
        =
        \frac{\lambda^2 \Cov^2(\mu(X), \alpha(X)|D(1)>D(0)) + \frac{\E[((1-\lambda) \: \E[\alpha(X)|D(1)>D(0)]  + \lambda \: \alpha(X))^2 \sigma^2(X)]}{p(1-p) \P^2(D(1)>D(0))}}{\E^2[\alpha(X)|D(1)>D(0)]}
    \end{align*}
    with
    \begin{align*}
        \frac{\partial}{\partial \lambda} L(\lambda) |_{\lambda = 0}
        =
        2 \frac{\E[(\alpha(X) - \E[\alpha(X)|D(1)>D(0)])\E[\alpha(X)|D(1)>D(0)] \sigma^2(X)]}{\E^2[\alpha(X)|D(1)>D(0)] p(1-p) \P^2(D(1)>D(0))}
    \end{align*}
    which is negative whenver
    \begin{align*}
        &
        &
        &\E[\alpha(X) \sigma^2(X)] - \E[\alpha(X)|D(1)>D(0)] \E[\sigma^2(X)] < 0
        \\
        &\Leftrightarrow
        &
        &\E\left[\frac{\alpha(X)}{\E[\alpha(X)]} \sigma^2(X)\right] - \E\left[\left(\frac{\alpha(X)}{\E[\alpha(X)]}\right)^2\right] \E[\sigma^2(X)] < 0
        \\
        &\Leftrightarrow
        &
        &\Cov\left(\frac{\alpha(X)}{\E[\alpha(X)]},\sigma^2(X)\right) + \E[\sigma^2(X)] - \Var\left(\frac{\alpha(X)}{\E[\alpha(X)]}\right) \E[\sigma^2(X)] - \E[\sigma^2(X)] < 0,
    \end{align*}
    in which case loss is lower for small $\lambda > 0$ relative to $\lambda=0$.
\end{proof}

\begin{proof}[Proof of \autoref{prop:generalization}]
    Follow from the general result in \autoref{prop:estimatedconsistent} and \autoref{prop:estimatedasymptotic} that allow weights to be estimated and does not assume binary instruments.
\end{proof}

\begin{proof}[Proof of \autoref{prop:generalcharacterization}]
    Since $\E[\varepsilon]=0$, $\Cov(h(Z,X), \varepsilon)=\E[h(Z,X) \varepsilon]$.
    Since also $Z$ independent of $(X,\varepsilon)$,
    $\E[h(Z,X) \varepsilon] = \E[\E[h(Z,X)|X,\varepsilon] \varepsilon] = \E[\E[h(Z,X)|X] \varepsilon]$.
    For this equation to hold for all conformal distributions of $(X,\varepsilon)$, including $\varepsilon = \E[h(Z,X)|X]$, we must have
    \begin{align}
        \label{eqn:conditionalmoment}
        \E[h(Z,X)|X] \equiv 0.
    \end{align}
    For $\Z = \{z_0,\ldots,z_J\}$,
    we can express
    \begin{align*}
        h(Z,X)
        =
        \sum_{j=0}^J \Ind_{Z=z_j} h_j(X).
    \end{align*}
    
    By \autoref{eqn:conditionalmoment} and independence of $Z$,
    \begin{align*}
        0
        =
        \E[h(Z,X)|X] 
        =
        \sum_{j=0}^J \P(Z=z_j) \: h_j(X).
    \end{align*}
    Thus,
    \begin{align*}
        h(Z,X)
        &=
        h(Z,X) - \E[h(Z,X)|X] \\
        &=
        \sum_{z \in \Z } (\Ind_{Z=z} - \P(Z=z)) h_z(X) \\
        &=
        \sum_{z \in \Z } (\Ind_{Z=z} - \P(Z=z)) (h_z(X) - h_{z_0}(X)) +  \underbrace{\sum_{z \in \Z } (\Ind_{Z=z} - \P(Z=z)) }_{=1-1=0} h_{z_0}(X)\\
        &=
        \sum_{z \in \Z \setminus \{z_0\}} (\Ind_{Z=z} - \P(Z=z)) (h_z(X) - h_{z_0}(X)).
    \end{align*}
    For the specific, orthonormal choice of basis, assume wlog that $z_j = j$ for all $j$, and note that
    \begin{align*}
        &\Ind_{z=j} - \P(Z=j)
        \\
        &=
        \Ind_{z=j} - \P(Z=j | Z \leq j) \: \Ind_{z \leq j}
        + \P(Z=j | Z \leq j) \: (1 - \Ind_{z > j}) - \P(Z=j)
        \\
        &=
        \Ind_{z=j} - \P(Z=j | Z \leq j) \: \Ind_{z \leq j}
        +  \P(Z=j) \frac{1 - \Ind_{z > j} - \P(Z \leq j)}{\P(Z \leq j)}
        \\
        &=
        \Ind_{z=j} - \P(Z=j | Z \leq j) \: \Ind_{z \leq j}
        -  \frac{\P(Z=j)}{\P(Z \leq j)} (\Ind_{z > j} - \P(Z > j))
        \\
        &= \Ind_{z=j} - \P(Z=j | Z \leq j) \: \Ind_{z \leq j} - \sum_{j' = j+1}^J (\Ind_{z = j'} - \P(Z  = j')).
    \end{align*}
    By induction over $j$ from $j=J$ to $1$, we can express all $\Ind_{z=j} - \P(Z=j)$ by linear combinations of $\Ind_{z=j} - \P(Z=j | Z \leq j)$ (with $j > 0$).
    Rescaling appropriately, we thus can find $w_j(X)$ with
    \begin{align*}
        h(Z,X)
        =
        \sum_{z \in \Z \setminus \{z_0\}} (\Ind_{Z=z} - \P(Z=z)) (h_z(X) - h_{z_0}(X))
        =
        \sum_{j=1}^J g_j(Z) w_j(X)
    \end{align*}
    with the above choice of $g_j$.
    Finally, note that independence of $Z$ and $(X,\varepsilon)$ and $\E[\Ind_{Z=z} - \P(Z=z)] = 0$ implies that $\Cov(h(Z,X), \varepsilon)=\E[h(Z,X) \varepsilon] = 0$ for all conformal distributions.
\end{proof}

\begin{proof}[Proof of \autoref{prop:estimatedconsistent} and \autoref{prop:estimatedasymptotic}]

    For equivalence, we first consider the case of controlling just for the weight $W = \hat{w}(X)$.
    Let
    \begin{align*}
        \dot{Y} &= Y - \sE[Y] - \frac{\sCov(Y,W)}{\sVar(W)} (W - \sE[W])
        \\
        \dot{D} &= D - \sE[D] - \frac{\sCov(D,W)}{\sVar(W)} (W - \sE[W])
        % \\
        % Z &= Z - \sE[Z] - \frac{\sCov(Z,W)}{\sVar(W)} (W - \sE[W])
    \end{align*}
    denote the residuals of $Y$ and $D$ after linearly controlling for $w(X)$ (and an intercept).
    Then
    \begin{align}
        \label{eqn:doubleexp}
        \sE[\dot{Y}] = \sE[W \: \dot{Y}] = \sE[\dot{D}] = \sE[W \: \dot{D}] = 0
    \end{align}
    and thus
    \begin{align*}
        \hat{\tau}_{\hat{w}} = \frac{\sCov^{\hat{w}}(\dot{Y},Z)}{\sCov^{\hat{w}}(\dot{D},Z)}
        =
        \frac{\sE[W \dot{Y} Z]}{\sE[W \dot{D} Z]}
        =
        \frac{\sCov(\dot{Y},W Z)}{\sCov(\dot{D},W Z)}
        =
        \hat{\tau}^{\hat{h}}.
    \end{align*}
    This yields equivalence without additional controls.
    With additional controls, the same argument extends by Frisch--Waugh--Lovell; indeed, we can first residualize the outcome and all other control variables $C_1,\ldots,C_k$ with respect to $W$  (and a constant) to obtain $\dot{Y},\dot{D}, \dot{C}_1,\ldots,\dot{C}_k$; then, residualized outcomes are
    \begin{align*}
        \ddot{Y} &= \dot{Y} - \hat{\phi}_2' \dot{C},
        &
        \ddot{D} &= \dot{D} - \hat{\phi}_1' \dot{C},
    \end{align*}
    and the above results extend, using linearity and \autoref{eqn:doubleexp} for $\dot{Y},\dot{D}, \dot{C}_1,\ldots,\dot{C}_k$.

    For the asymptotic distribution,
    given equivalence we can focus on the representation in terms of simple averages, for which, writing $\alpha'(x) = p (1-p) \alpha(x)$
    \begin{align}
        \label{eqn:bigdiff}
        \begin{split}
        \frac{\sE[W \ddot{Y} Z]}{\sE[W \ddot{D} Z]}
        -
        \tau^{\text{CF}}_{\hat{w}}
        &=
        \frac{\sum_{j=1}^k \sE[W \ddot{Y} Z|j(i) {=} j]}{\sum_{j=1}^k \sE[W \ddot{D} Z|j(i) {=} j]} - \frac{\sum_{j=1}^k \E[\alpha'(X) \: \hat{w}_{-j}(X) \: \tau(X) |\hat{w}_{-j}]}{\sum_{j=1}^k \E[\alpha'(X) \: \hat{w}_{-j}(X)  |\hat{w}_{-j}]}
        \\
        &=
        \frac{\sum_{j=1}^k (\sE[W \ddot{Y} Z|j(i) {=} j] - \E[\alpha'(X) \: \hat{w}_{-j}(X) \: \tau(X) |\hat{w}_{-j}])}{\sum_{j=1}^k \sE[W \ddot{D} Z|j(i) {=} j]}
        \\
        &\peq   
        - 
        \frac{\sum_{j=1}^k (\sE[W \ddot{D} Z|j(i) {=} j] - \E[\alpha'(X) \: \hat{w}_{-j}(X)|\hat{w}_{-j}])}{\sum_{j=1}^k \sE[W \ddot{D} Z|j(i) {=} j]}
        \tau^{\text{CF}}_{\hat{w}}
        \\
        &=
        \frac{\sum_{j=1}^k (\sE[W (\ddot{Y} {-} \ddot{D} \tau_w) Z|j(i) {=} j] - \E[\alpha'(X) \: \hat{w}_{-j}(X) \: (\tau(X) {-} \tau_w) |\hat{w}_{-j}])}{\sum_{j=1}^k \sE[W \ddot{D} Z|j(i) {=} j]}
        \\
        &\peq   
        - 
        \frac{\sum_{j=1}^k (\sE[W \ddot{D} Z|j(i) {=} j] - \E[\alpha'(X) \: \hat{w}_{-j}(X)|\hat{w}_{-j}])}{\sum_{j=1}^k \sE[W \ddot{D} Z|j(i) {=} j]}
        (\tau^{\text{CF}}_{\hat{w}} - \tau_w)
        \end{split}
    \end{align}
    Write
    \begin{align*}
        Y^* &= Y - \psi_2 - w(X) \: \omega_2 - C'\varphi_2,
        &
        D^* &= D - \psi_1 - w(X) \: \omega_1 - C'\varphi_1
    \end{align*}
    for the residuals in the population regression of $Y$ and $D$ on a constant, the weight $w(X)$, and an additional, fixed-dimension control vector $C=c(X)$ of controls,
    and 
    \begin{align*}
        \ddot{Y} &= Y - \hat{\psi}_2 - W \: \hat{\omega}_2 - C'\hat{\varphi}_2,
        &
        \ddot{D} &= D - \hat{\psi}_1 - W \: \hat{\omega}_1 - C'\hat{\varphi}_1
    \end{align*}
    for the empirical analogue (using estimated weights).
    By \autoref{lem:weightedclt} applied for $W_i$ and $W_i'=W_i^2$ (for the covariances involving $w(X)$) and the law of large numbers (for other covariances),
    where we assume for simplicity that the $W_i$ are bounded and that the population regression is well conditioned,
    \begin{align}
        \label{eqn:olsconsistent}
        \begin{pmatrix}
            \hat{\psi}_2 \\
            \hat{\omega}_2 \\
            \hat{\varphi}_2
        \end{pmatrix}
        & \cp
        \begin{pmatrix}
            \psi_2 \\
            \omega_2 \\
            \varphi_2
        \end{pmatrix},
        &
        \begin{pmatrix}
            \hat{\psi}_1 \\
            \hat{\omega}_1 \\
            \hat{\varphi}_1
        \end{pmatrix}
        & \cp
        \begin{pmatrix}
            \psi_1 \\
            \omega_1 \\
            \varphi_1
        \end{pmatrix}.
    \end{align}
    As a preliminary calculation (for $p = \E[Z]$),
    \begin{align*}
        &\sqrt{n} (\overbrace{\sE[W \ddot{Y} Z]}^{\mathclap{=\sE[W \ddot{Y} (Z-p)]}}
        - \sE[W Y^* (Z-p)])
        \\
        &=
        - \sqrt{n} \sE[W ((\hat{\psi}_2 - \psi_2) + C'(\hat{\varphi}_2 - \varphi_2)
        + W (\hat{\omega}_2 - \omega_2)
        + (W - w(X)) \omega_2) (Z-p) ]
        \\
        &=
        - \sqrt{n} ((\hat{\psi}_2 - \psi_2) \sE[W (Z-p)] + (\hat{\varphi}_2 - \varphi_2)'\sE[c(X) W (Z-p)])
        \\
        &\peq
        -
        (\hat{\omega}_2 - \omega_2)
        \sqrt{n} \sE[W^2 (Z-p)]
        -  \omega_2 \sqrt{n} \underbrace{\sE[W (W - w(X)) (Z-p) ]}_{\mathclap{
            =
            \sE[(W^2 - w^2(X)) (Z-p) ]
            -
            \sE[(W - w(X)) (Z-p) w(X) ]
        }}.
    \end{align*}
    By independence of $Z$ and $X$ and \autoref{lem:weightedclt} applied to $W$ and $W'=W^2$,
    \begin{align*}
        &\sqrt{n} \sE[W (Z-p)],
        &
        &\sqrt{n}\sE[W c(X) (Z-p)],
        &
        &\sqrt{n} \sE[W^2 (Z-p)]
    \end{align*}
    are all asymptotically Normal with zero mean and finite variance.
    By \autoref{eqn:residualsmall} in the proof of \autoref{lem:weightedclt}, applied for $W$ and $W'=W^2$ as well as mean-zero residuals $Z-p$ and $(Z-p)w(X)$, respectively,
    \begin{align*}
        \sqrt{n} \sE[(W^2 - w^2(X)) (Z-p) ]
        &\cp 0,
        &
        \sqrt{n} \sE[(W - w(X)) (Z-p) w(X) ]
        &\cp 0.
    \end{align*}
    With \autoref{eqn:olsconsistent} if follows that $\sqrt{n} (\sE[W \ddot{Y} Z]
    - \sE[W Y^* (Z-p)]) \cp 0$; for $D$ instead of $Y$ similarly $\sqrt{n} (\sE[W \ddot{D} Z]
    - \sE[W D^* (Z-p)]) \cp 0$.
    Using this asymptotic equivalence, $\hat{\tau}_{\hat{w}}=\frac{\sE[W \ddot{Y} Z]}{\sE[W \ddot{D} Z]}$ is consistent by \autoref{lem:weightedclt} (applied to numerator and to denominator separately) for both $\tau_w$ and for $\tau^{\text{CF}}_{\hat{w}}$,
    so we also have that
    \begin{align*}
        \tau^{\text{CF}}_{\hat{w}} \cp \tau_w.
    \end{align*}
    Hence,
    writing
    \begin{align*}
        \varepsilon
        =
        Y^* - D^* \tau_w
    \end{align*}
    and noting that
    $\E[\varepsilon (Z-p)|X=x] = \alpha'(X) (\tau(X) - \tau_w)]$ and $\E[D^* (Z-p)|X=x] = \alpha'(X)$,
    by \autoref{lem:weightedclt} we find that
    \begin{align*}
        \sqrt{n} \sum_{j=1}^k {(\sE[W (\ddot{Y} {-} \ddot{D} \tau_w) Z|j(i) {=} j]
        \atop
        - \E[\alpha'(X) \: \hat{w}_{-j}(X) \: (\tau(X) {-} \tau_w) |\hat{w}_{-j}])}
        &\cd \N(0,\Var(w(X) \varepsilon(Z-p)))
        \\
        \sum_{j=1}^k \sE[W \ddot{D} Z|j(i) {=} j]
        &\cp \E[w(X) \alpha'(X)]
        \\
        \sqrt{n} \frac{\sum_{j=1}^k (\sE[W \ddot{D} Z|j(i) {=} j] - \E[\alpha'(X) \: \hat{w}_{-j}(X)|\hat{w}_{-j}])}{\sum_{j=1}^k \sE[W \ddot{D} Z|j(i) {=} j]}
        (\tau^{\text{CF}}_{\hat{w}} - \tau_w)
        &\cp 0.
    \end{align*}
    By \autoref{eqn:bigdiff},
    it follows that
    \begin{align*}
        \sqrt{n} (\hat{\tau}_{\hat{w}} - \tau_{\hat{w}}^{\text{CF}})
        \cd \N\left(0,\frac{\Var(w(X) \varepsilon(Z-p))}{p^2 (1-p)^2 \E^2[w(X) \alpha(X)]}\right)
    \end{align*}
    where
    \begin{align*}
        \Var(w(X) \varepsilon (Z-p))
        &= \E[\Var(w(X) \varepsilon (Z-p)|Z)]
        \\
        &=p(1-p) ((1-p)\E[w^2(X) \varepsilon^2|Z=1] + p \E[w^2(X) \varepsilon^2|Z=0])\\
        &=p(1-p)
            \E[w^2(X) ((1-p)\E[\varepsilon^2|X,Z=1]+ p\E[\varepsilon^2|X,Z=0]) ]
        .\qedhere
    \end{align*}
\end{proof}

\begin{proof}[Proof of \autoref{prop:robustse}]
    Consistency follows as in the proof of \autoref{prop:estimatedasymptotic} from consistency of the OLS coefficients for the regressions of $Y$ and $D$ on a constant, the weight, and controls, from \autoref{lem:weightedclt} applied for powers of the weight (where we assume that the weight is bounded), and from the conventional law of large numbers (for the empirical variance of the instrument).
\end{proof}

\section{Auxiliary Derivations for Simulation}
\label{sec:sim_deriv}
For the simulation, we assume that
\begin{align*}
    &\begin{pmatrix}
        \delta \\
        \varepsilon \\
        \tau
    \end{pmatrix}
    \sim
    \N\left(
        \0,
        \begin{pmatrix} 
            1 & \rho_{\delta\varepsilon} & \rho_{\delta\tau}\sigma_{\tau} \\ 
            \rho_{\delta\varepsilon} & 1  & \rho_{\tau\varepsilon}\sigma_{\tau}\\ \rho_{\delta\tau}\sigma_{\tau}  & \rho_{\tau\varepsilon}\sigma_{\tau} & \sigma_{\tau}^2
        \end{pmatrix}
    \right)
\end{align*}
and
\begin{align*}
    D(1) &= \Ind(\Phi^{-1}(\delta) > S_{NT}),
    &
    D(0) &= \Ind(\Phi^{-1}(\delta) > 1 - S_{AT}),
    \\
    Y(d) &= d \: \tau + (1+\zeta\delta) \varepsilon,
    &
    X &= \delta + \eta
\end{align*}
where $\eta \sim \N(0,\sigma_{\eta}^2)$ independently of $\delta,\varepsilon,\tau$, and $\Phi(\cdot) = \P(\N(0,1) \leq \cdot)$.
Further, since $X | \delta \sim \N(\delta,\sigma^2_\eta)$ and $\delta \sim \N(0,1)$, we have that
\begin{align*}
    \delta | X \sim \N\left(
        \frac{X}{1 + \sigma^2_\eta},
        \frac{\sigma^2_\eta}{1+ \sigma^2_\eta}
    \right).
\end{align*}
Also, $\E[\tau|\delta] = \rho_{\delta \tau} \sigma_\tau \delta$, so
\begin{align*}
    \tau(x)
    &=
    \E[Y(1) - Y(0) |X=x]
    =
    \E[\tau |X=x]
    =
    \E[\E[\tau|X,\delta] |X=x] \\
    &= \E[\E[\tau|\delta,\eta] |X=x]
    = \E[\E[\tau|\delta] |X=x]
    = \frac{\rho_{\delta \tau} \sigma_\tau x}{1 + \sigma^2_\eta},
    \\
    \alpha(x)
    &=
    \P(D(1) > D(0) | X=x)
    =
    \P(\Phi(S_{NT}) < \delta \leq \Phi(1-S_{AT})|X=x)
    \\
    &=
    \P\left(\Phi(S_{NT}) < 
    \N\left(
        \frac{x}{1 + \sigma^2_\eta},
        \frac{\sigma^2_\eta}{1+ \sigma^2_\eta}
    \right)
    \leq \Phi(1-S_{AT})\right)
    \\
    &=
    \Phi\left( 
        \sqrt{1+ \sigma^{-2}_\eta} \: \Phi^{-1}(1-S_{AT}) - \frac{x}{\sqrt{\sigma_\eta^2 + \sigma^4_\eta}}
        \right)
    -
    \Phi\left( 
        \sqrt{1+ \sigma^{-2}_\eta} \: \Phi^{-1}(S_{NT}) - \frac{x}{\sqrt{\sigma_\eta^2 + \sigma^4_\eta}}
        \right).
\end{align*}
    \end{vappendix}

\end{document}